\documentclass[a4paper,11pt]{scrartcl}
\usepackage[a4paper,
inner=2.3cm, 
outer=2.3cm,
top=2.8cm,
bottom=2.8cm]{geometry}
\usepackage{indentfirst}
\usepackage{enumitem}

\usepackage{hyperref}
\hypersetup{
	colorlinks,
	linkcolor={red!50!black},
	citecolor={blue!50!black},
	urlcolor={blue!80!black}
}

\usepackage{amsmath}
\usepackage{amsfonts}
\usepackage{mathtools}
\usepackage{amssymb}
\usepackage{tikz-cd}
\usepackage{dsfont}
\usepackage{bm}

\newcounter{dummy} 
\numberwithin{dummy}{section}

\usepackage[hyperref,amsmath,thmmarks]{ntheorem}

\theoremstyle{plain}
\theoremseparator{.}
\theoremsymbol{\scalebox{0.7}{$ \square $}}
\theorembodyfont{}
\newtheorem{proposition}[dummy]{Proposition}

\theoremstyle{plain}
\theoremseparator{.}
\theoremsymbol{\scalebox{0.7}{$ \square $}}
\theorembodyfont{}
\newtheorem{lemma}[dummy]{Lemma}

\theoremstyle{plain}
\theoremseparator{.}
\theoremsymbol{\scalebox{0.7}{$ \square $}}
\theorembodyfont{}
\newtheorem{remark}[dummy]{Remark}

\theoremstyle{plain}
\theoremseparator{.}
\theoremsymbol{\scalebox{0.7}{$ \square $}}
\theorembodyfont{}

\theoremstyle{plain}
\theoremseparator{.}
\theoremsymbol{\scalebox{0.7}{$ \square $}}
\theorembodyfont{}
\newtheorem{theorem}[dummy]{Theorem}

\newtheorem{corollary}[dummy]{Corollary}

\theoremstyle{nonumberplain}
\theoremsymbol{$ \blacksquare $}
\theoremheaderfont{\scshape}
\newtheorem{proof}{Proof}

\allowdisplaybreaks

\usepackage{xcolor}
\definecolor{darkgreen}{rgb}{0.0,0.5,0.0}

\usepackage{subcaption}
\usepackage{caption}

\usepackage{tikz}

\newcommand{\F}[1][R]{\mathbb{#1}} 
\newcommand{\setsep}{\ \middle|\ } 
\newcommand{\Eucprod}[2]{\left\langle #1,#2 \right\rangle} 
\newcommand{\Ltwonorm}[1]{\left\Vert #1 \right\Vert} 

\newcommand{\SE}{\mathrm{SE}}
\newcommand{\SO}{\mathrm{SO}}
\newcommand{\so}{\mathfrak{so}}
\newcommand{\Sp}{\mathop{\mathrm{Span}}}

\newcommand{\LtwoStwo}{L^{2} \left(S^{2} \right)}

\newcommand{\re}{\mathfrak{Re}}
\newcommand{\im}{\mathfrak{Im}}
\newcommand{\SubmerSE}{\eta_{\SE(2)}}
\newcommand{\SubmerSO}{\eta_{\SO(3)}}
\newcommand{\GrTr}{f}
\newcommand{\MaxTrans}{T_{\mathrm{max}}}
\newcommand{\ImageSpace}{L^{2}_{c} \left( \F^{2} \right)}

\usepackage[backend=bibtex, style=ieee, citestyle=numeric]{biblatex}
\usepackage{cleveref}

\title{Compactification of the Rigid Motions Group\\in Image Processing}
\subtitle{}
\author{Tamir Bendory, Ido Hadi, Nir Sharon}

\begin{document}
	
	\maketitle
	
	\begin{abstract}
		Image processing problems in general, and in particular in the field of single-particle cryo-electron microscopy, often require considering images up to their rotations and translations. Such problems were tackled successfully when considering images up to rotations only, using quantities which are invariant to the action of rotations on images. Extending these methods to cases where translations are involved is more complicated.  Here we present a computationally feasible and theoretically sound approximate invariant to the action of rotations and translations on images. It allows one to approximately reduce image processing problems to similar problems over the sphere, a compact domain acted on by the group of 3D rotations, a compact group. We show that this invariant is induced by a family of mappings deforming, and thereby compactifying, the group structure of rotations and translations of the plane, i.e., the group of rigid motions, into the group of 3D rotations. Furthermore, we demonstrate its viability in two image processing tasks: multi-reference alignment and classification. To our knowledge, this is the first instance of a quantity that is either exactly or approximately invariant to rotations and translations of images that both rests on a sound theoretical foundation and also applicable in practice.
	\end{abstract}

	\section{Introduction}
	In image processing applications, there is often a need to consider images up to rotation, up to translation or up to both. 
	Two examples of this are found in multi-reference alignment (MRA) and image classification problems. In one of its more well-studied forms, the aim in MRA is to estimate an image up to rotation from a dataset of noisy rotated copies of it \cite{Bandeira2017,Bandeira2020,Ma2020,Perry2019}. 
	In classification problems of the type we consider in this paper, the aim is to estimate the similarity up to rotation between every pair of images in a noisy dataset. 
	In both cases, rather than using rotations, the problems can be considered over translations or both rotations and translations.
	
	Rotations and translations of the plane form a group, $ \SE(2) $, usually referred to as rigid motions or the special Euclidean group. This group is non-compact with an algebraic structure of a semi-direct product of the translation and rotation groups. The group structure enables one to formalize the notion of considering an image up to rotation, translation, or both in group-theoretic terms. In particular, an orbit of an image consists of all shifted and rotated copies of the image. Thus, the group action induces an equivalence relation on any set of images where the orbits act as the equivalence classes. Therefore, we are often required to consider only the orbits' representatives. 
	
	Using this group-theoretic formalism, an MRA problem amounts to estimating an orbit, or more accurately a representative of an orbit, from noisy samples of members of that orbit. Classification problems amount to estimating the minimal distance between orbits from noisy samples of members of these orbits.
	Formalizing these problems in this way enables well-developed concepts from group theory to be brought to bear on them. 
	In particular, it leads to a wider applicability of what we refer to here as the invariants approach to orbit-related problems. 
	Central to this approach are quantities that are invariant to the action of a group on an image.
	Formally, these invariants are operators on images that are constant on every orbit.
	It is trivial to find such invariant operators. For example, every operator that is constant on images is also invariant to any action by any group on it.
	However, combining insights from group theory and harmonic analysis on groups, a wide variety of non-trivial invariants were derived which are not only constant on every orbit, but also uniquely determine an orbit \cite{Kakarala1992,Bendory2018,Kakarala2010,Ma2020,Bandeira2017,Bandeira2020,Perry2019}. These invariants can be thought of as injective functions on the set of orbits of a group action and therefore essentially characterize the set of orbits of the action of a group on an image.
	
	Using these invariants yielded a computationally inexpensive solution to the MRA \cite{Ma2020} and classification \cite{Zhao2014} problems with respect to the rotation group. 
	In both cases, an $ \SO(2) $-invariant was evaluated on every image in the dataset, resulting in an $ \SO(2) $-invariant representation of every image in the dataset. 
	In \cite{Ma2020}, the MRA problem  was handled by using these representations to estimate the $ \SO(2) $-invariant representation of the underlying image, and then inverted to yield an estimate of the image itself, up to rotation.
	In \cite{Zhao2014}, the classification problem was tackled by taking the distance between pairs of images to be the distance between their $ \SO(2) $-invariant representations. 
	
	In both \cite{Ma2020} and \cite{Zhao2014}, as well as in our work, the motivating real-world application is single-particle reconstruction in cryogenic electron microscopy (cryo-EM). In cryo-EM, the goal is to reconstruct the three-dimensional structure of biological macromolecules from images generated by an electron microscope.
	While  \cite{Ma2020} and \cite{Zhao2014} handled MRA and classification, respectively, up to rotation only, in cryo-EM, as in other applications, there is often a need to consider images up to both rotations and translations \cite{Bendory2020}.
	In particular, the images used to estimate the three-dimensional structure can be considered up to rotations and translations.
	One would therefore like to generalize the approach of  \cite{Ma2020} and \cite{Zhao2014} by using an orbit-characterizing $ \SE(2) $-invariant.
	
	Unfortunately, applying the invariant approach to MRA and classification up to both rotation and translations is considerably harder. 
	Deriving computationally feasible orbit-characterizing invariants for the rigid motions group $ \SE(2) $ is hard, since this group is hard to analyze algebraically for two reasons. 
	First, it is non-compact and harmonic analysis over non-compact groups is considerably more complicated than over compact groups, such as $ \SO(n) $, the rotation group of $ \F^{n} $. 
	Second, as a semi-direct product it has a rather complicated algebraic structure.
	Thus, so far as we know, no computationally feasible orbit-characterizing invariants with sound theory behind them are known for $ \SE(2) $.
	
	Since working over compact group is simpler, both theoretically and in practice, we seek to compactify orbit-related problems over the rigid motions group. Specifically, we wished to approximately reduce orbit-related image processing problems over $ \SE(2) $ to similar problems over compact groups. 
	This approach has worked in the past. A heuristically derived approximate $ \SE(2) $-invariant was applied by Kondor \cite{Kondor2007} to an image classification problem. The invariant was computed by projecting images from the plane onto the sphere and computing an $ \SO(3) $-invariant on the sphere.
	Similarly, in \cite{Sharon2018} a synchronization problem over $ \SE(2) $ itself, a problem where the objective was to estimate elements of $ \SE(2) $, was approximately reduced to a problem over the compact group $ \SO(3) $ (for more information on synchronization, see \cite{Carlone2015,Rosen2019,Tron2016}).
	
	Here we extend the theoretical approach of \cite{Sharon2018} to derive and theoretically justify the use of the approximate invariant used by \cite{Kondor2007}.
	We show the projection used by \cite{Kondor2007} is induced by contraction maps, a family of mappings from $ \SE(2) $ to $ \SO(3) $ that rigorously  formalize the notion of continuously deforming one group into another. 
	We then prove this projection maps an orbit of the action of $ \SE(2) $ on images to approximately the same orbit of the action of $ \SO(3) $ on the space of functions on the sphere.
	We then apply our approximate $ \SE(2) $-invariant to MRA and classification problems over $ \SE(2) $ and show they provide viable solutions on simulated data, thus demonstrating its potential.
	Our MATLAB code to reproduce all numerical experiments is available at \url{https://github.com/idohadi/CompactificationImageProcessing}.
	To our knowedge, we provide here the first instance of a quantity that is either exactly or approximately $ \SE(2) $-invariant on images that both rests on a sound theoretical foundation and also computable in practice.

	\subsection{Structure of the paper}
	We begin in \Cref{CompactifiactionTheorySection} by constructing a projection of compactly supported functions on the plane onto functions on the sphere
	and use it to define an approximate $ \SE(2) $ invariant.
	After deriving the theory, we discuss in \Cref{CompactifiactionComputationSection} several computational aspects of projecting discrete images onto the sphere and of calculating their spherical bispectrum. 
	In particular, we demonstrate numerically that our approximate $ \SE(2) $ invariant is deserved to be named so.
	In the rest of the paper, we apply our approximate $ \SE(2) $ invariant to two problems of interest. 
	First, in \Cref{MRASection}, we flesh out a modified invariants approach to MRA using our approximate $ \SE(2) $ invariant.
	We perform several numerical experiments demonstrating our approach is viable.
	Second, in \Cref{ClassificationSection} we discuss image classification up to rotation and translation. 
	We use our approximate $ \SE(2) $ invariant to derive an approximately $ \SE(2) $-invariant measure of the distance between images. We then demonstrate by numerical experiments that our approach is able to identify similar images up to rotations and translations and can do so even for large translation sizes.
	For both MRA and classification, we show that our approach outperforms methods that take only rotations into account.

	\section{Compactifiaction of functions on the plane: theoretical aspects}
	\label{CompactifiactionTheorySection}
	
	Given a function on the plane, we wish to define a function on the sphere which preserves, in some sense, two key properties of the original function. First, the original function needs to be recoverable from its spherical counterpart. Second, the group action of $ \SE(2) $ on the original function needs to be reflected, in some sense, in the group action of $ \SO(3) $ on its spherical counterpart. 
	
	In this section, we construct a mapping between functions on the plane and functions on the sphere which has these two properties. Center to this construction is a special family of maps, called contraction maps, between the Lie algebras of $ \SE(2) $ and $ \SO(3) $. 
	Because the subgroup structure of a Lie group is recapitulated in its Lie algebra, contraction maps provide a sense in which one can continuously deform one Lie group into another \cite{Dooley1984}. We use these maps to construct and study the properties of the mapping described above.
	
	\Cref{NotationSection} introduces the notation used throughout the paper. In particular, we introduce explicitly the group-theoretic formalization of considering images up to rotations, translations or both.
	\Cref{ContractionMapsSection} introduces contraction maps in general and the family of contraction maps between $ \SE(2) $ and $ \SO(3) $ we use in particular. In \Cref{ProjectionToSphereSection}, we use this family of contraction maps to induce a map between functions on the plane and functions on the sphere. 
	We also provide an explicit formula for this map and prove that the original function on the plane can be recovered from its spherical counterpart. 
	In \Cref{ApproximationTheoremSection}, we study the manner in which the group action on the sphere is related to the original group action on the plane. In particular, we show that a member of the orbit of the original image is mapped to a neighborhood of the orbit of its spherical counterpart. This neighborhood gets at most exponentially larger with the size of the translational part of the element of $ \SE(2) $ acting on the original function. 
	In the final section, \Cref{SphericalBispectrumSection}, we introduce the spherical bispectrum of a function on the sphere, an invariant to the action of $ \SO(3) $ that uniquely determines functions on the sphere, up to an action of $ \SO(3) $. Taken together, we show that when it is evaluated on images projected onto the sphere, the spherical bispectrum is an approximate $ \SE(2) $ invariant.

	\subsection{Notation}
	\label{NotationSection}
	
	Throughout the paper, lowercase letters are scalars, lowercase boldface letters are column vectors and uppercase boldface letters are matrices. 
	The standard Euclidean norm is $ \Ltwonorm{\cdot} = \Ltwonorm{\cdot}_{2} $, where $ \Ltwonorm{\mathbf{x}}_{2} = \sqrt{\sum_{j=1}^{n} \left|x_{j}\right|^{2}} $ for $ \mathbf{x} = \left( x_{1}, x_{2}, \dots, x_{n} \right)^{\top} $ in $ \F^{n} $ or in $ \F[C]^{n} $. 
	The open $ r $-ball in $ \F^{2} $ is $ B_{r} \coloneqq \left\{ \mathbf{x} \in \F^{2} \setsep \Ltwonorm{\mathbf{x}} < r \right\} $. The closed $ r $-ball is denoted by $ \overline{B}_{r} \coloneqq \left\{ \mathbf{x} \in \F^{2} \setsep \Ltwonorm{\mathbf{x}} \le r \right\} $.
	Finally, $ \partial B_{r} \coloneqq \left\{ \Ltwonorm{\mathbf{x}} \in \F^{2} \setsep \Ltwonorm{\mathbf{x}} = r \right\} $.
	
	The elementwise complex conjugate of a $ n \times m $ matrix $ \mathbf{M} $ is denote by $ \mathbf{M}^{*}  $. Thus, for example $ \mathbf{x}^{*} = \left( x_{1}^{*}, x_{2}^{*}, \dots,x_{n}^{*}\right)^{\top} $ for every $ \mathbf{x} \in \F[C]^{n} $.
	The transpose of an $ n \times m $ matrix $ \mathbf{M} $ is denoted by $ \mathbf{M}^{\top} $ and its adjoint is denoted by $ \mathbf{M}^{\dag} $. Note that $ \mathbf{M}^{\dag} = \left(\mathbf{M}^{*}\right)^{\top} $.
	
	The sphere is the unit $ 2 $-sphere, $ S^{2} \coloneqq \left\{\mathbf{x} \in \F^{3} \setsep \Ltwonorm{\mathbf{x}} = 1 \right\} $.
	The north pole of the sphere is denoted by $\mathbf{n} \coloneqq (0, 0, 1)^{\top} $.
	The $ d_{S^{2}} \left(\mathbf{x}_{1}, \mathbf{x}_{2}\right) = \arccos\left(\mathbf{s}_{2}^{\top} \mathbf{s}_{1}\right) $ is the great-circle distance metric on the sphere.
	$ \SO(3) $ acts on a function on the sphere $ f :S^{2} \to \F[C] $ by $ R \bullet f (\mathbf{s}) = f \left( R^{\top} \mathbf{s}\right) $.
	
	$ \SO(2) $ is the group of rotations of $ \F^{2} $ and $ \SO(3) $ is the group of rotations of $ \F^{3} $ or equivalently of $ S^{2} $. 
	The group of rigid motions, also referred to here as the group of rotations and translations of the plane, is denoted by $ \SE(2) = \F^{2} \rtimes \SO(2) $. 
	If $ g = \left(\mathbf{x}, R\right) \in \SE(2) $, we refer to $ \mathbf{x} $ as the translational part of $ g $ and to $ R $ as the rotational part.
	
	As we mentioned in the introduction, rotations and translations of the plane form a group, enabling one to formalize the notion of considering an image up to rotation, translation or both in group-theoretic terms.
	If $ G $ is a (multiplicative) group, we say it acts on a set $ X $ if every $ g \in G $ induces a map $ X \to X $ denoted by $ x \mapsto g \bullet x $ such that $  g_{1} \bullet \left( g_{2} \bullet x \right) = \left( g_{1} g_{2} \right) \bullet x $ for all $ g_{1}, g_{2} \in G $ and $ x \in X $.
	An orbit of the action of $ G $ on $ X $ is the set $ G \bullet x = \left\{ g \bullet x \setsep g \in G \right\} $ of all elements of $ X $ that can be reached by applying $ G $ to $ x $.
	It is easy to prove that $ G \bullet x_{1} = G \bullet x_{2} $ for $ x_{1},x_{2} \in X $ if and only if $ x_{1} = g \bullet x_{2} $ for some $ g \in G $. Therefore, the action of a group on a set induces an equivalence relation on the set with the orbits  functioning as the equivalence classes.
	Considering an image up to rotation, translation or both amounts to considering it as a representative of its orbit with respect to the action of the rotation group, the translation group or the translation and rotation group, respectively, on the space of images.
	
	In order to formalize the latter statement, we model images, as is conventional, as compactly supported, square-integrable real-valued functions on the plane and denote by $ \ImageSpace $ the space of such functions.
	The rotation group of the plane $ \SO(2) $ acts on $ \ImageSpace $ by $ f \mapsto R \bullet f $, 
	where $ R \bullet f (\mathbf{x}) = f \left( R^{\top} \mathbf{x} \right) $. 
	The (additive) translation group of the plane $ \F^{2} $ acts on $ \ImageSpace $ by $ f \mapsto \mathbf{b} \bullet f $, where $ \mathbf{b} \bullet f (\mathbf{x}) = f \left( \mathbf{x} - \mathbf{b} \right) $. 
	The group of translations and rotations of the plane $ \SE(2) $, usually referred to as the group of rigid motions of the plane or the special Euclidean group of the plane, has a more complicated algebraic structure. It is the semi-direct product of the translation group and the rotation group, $ \SE(2) = \F^{2} \rtimes \SO(2) $, and so  its group operation is $ (\mathbf{b}_{1}, R_{1})  (\mathbf{b}_{2}, R_{2}) = (\mathbf{b}_{1} + R_{1} \mathbf{b}_{2}, R_{1} R_{2}) $.
	This group acts on $ \F^{2} $ by $ (\mathbf{b}, R) \bullet \mathbf{x} =  R \mathbf{x} + \mathbf{b} $
	and on on 
	$ C_{c} \left( \F^{2} \right) $ by 
	$ f \mapsto \left(\mathbf{b}, R\right) \bullet f $, 
	where $ \left(\mathbf{b}, R\right) \bullet f (\mathbf{x}) = f \left( (\mathbf{b}, R)^{-1} \mathbf{x} \right) $.
	The inverse of $ (\mathbf{b}, R) \in \SE(2) $ can readily be shown to be $ \left(\mathbf{b}, R\right)^{-1} = \left( -R^{\top} \mathbf{b}, R^{\top} \right) $.

	\subsection{Contraction maps}
	\label{ContractionMapsSection}
	A popular construction in physics obtains a Lie algebra as the limit of a sequence of Lie algebras. This construction was first made rigorous in a landmark paper \cite{Inonu1953}. Dooley \cite{Dooley1984} gave a coordinate-free formulation of this construction using contraction maps. This section follows his approach.
	
	Consider a semisimple compact Lie group $ G $ and its Lie algebra $ \mathfrak{g} $ with Lie brackets $ \left[ \cdot, \cdot \right] $. Let $ G_{0} \le G $ be a continuous subgroup of $ G $ and denote its Lie algebra by $ \mathfrak{g}_{0} $. Since $ G_{0} \le G $, Lie correspondence ensures that $ \mathfrak{g}_{0} $ is a Lie subalgebra of $ \mathfrak{g} $. 
	If $ \mathfrak{g} = \mathfrak{g}_{0} \oplus V $, where
	\begin{equation*}
		\left[\mathfrak{g}_{0}, V\right] 
		= \left\{ \left[g_{0}, v \right]  \setsep g_{0} \in \mathfrak{g}_{0}, v \in V \right\}
		\subset V ,
	\end{equation*}
	define for every $ \lambda > 0 $ the contraction map of $ G $ with respect to $ G_{0} $:
	\begin{equation*}
		\label{AbstractContractionMap}
		\Psi_{\lambda} : V \rtimes G_{0} \to G \mbox{ such that }
		\left( v, g_{0} \right) \mapsto \exp_{G} \left( \frac{v}{\lambda}\right) g_{0}.
	\end{equation*}
	As noted in \cite{Dooley1984}, treated as distinct groups, both $ V \rtimes G_{0} $ and $ G $ have $ \mathfrak{g} $ as the underlying vector space of their Lie algebra. Yet, one can obtain the Lie brackets of $ V \rtimes G_{0} $ by transforming the Lie brackets of $ G $ using $ \Psi_{\lambda} $ and taking the limit $ \lambda \to \infty $. In this sense, the family of the contraction maps $ \left\{ \Psi_{\lambda} \setsep \lambda>0\right\} $ enables one to continuously deform one Lie algebra to produce another in the limit.
	
	In our case, $ G = \SO(3) $ and $ G_{0} = \SO(2) $, where $ \so (3) $ is the space of skew-symmetric $ 3 \times 3 $ real matrices equipped with the commutator as the Lie bracket; that is, the space of real $ 3  \times 3 $ matrices satisfying $ S^{\top} = - S^{\top} $ with the Lie bracket $ \left[ S_{1}, S_{2} \right] = S_{1} S_{2} - S_{2} S_{1} $.  Therefore, as a vector space $ \so (3) = \Sp \left\{ S_{i,j} \setsep 1 \le i < j \le 3 \right\} $, where $ S_{i,j} $ is a real skew-symmetric $ 3 \times 3 $ matrix with $ 1 $ and $ -1 $ at the $ (i,j) $ and $ (j,i) $ elements, respectively, and zero otherwise. Denoting $ S_{1} = S_{1,2} $, $ S_{2} = S_{1,3} $ and $ S_{3} = S_{2,3} $, one can easily show that 
	\begin{equation*}
		\left[ S_{1}, S_{2} \right] = - S_{3}, \ 
		\left[ S_{2}, S_{3} \right] = - S_{1} 
		\mbox{ and }
		\left[ S_{1}, S_{3} \right] = S_{2}.
	\end{equation*}
	Thus, $ \mathfrak{g}_{0} = \so(2) \cong \Sp \left\{ S_{1}\right\} $ and $ V \cong \Sp \left\{ S_{2}, S_{3} \right\} \cong \F^{2} $. Therefore, the contraction map of $ \SO(3) $ with respect to $ \SO(2) $ is 
	\begin{equation}
		\label{ContractionMap}
		\Psi_{\lambda} : \SE(2) \to \SO(3) 
		\mbox{ such that } 
		\left( \mathbf{x}, R \right) \mapsto \exp_{\SO(3)} \left( \frac{\mathbf{x}}{\lambda}\right) R.
	\end{equation}
	Crucially for what follows, this contraction map is smooth, because the exponential map and the group operation map $ \SO(3) \times \SO(3) \to \SO(3) $ are all smooth.
	
	In \cite{Sharon2018}, the contraction map \eqref{ContractionMap} was used to reduce a synchronization problem over $ \SE(n) $ to a synchronization problem over $ \SO(n) $. In particular, \cite[Prop. 3.4]{Sharon2018} shows that \eqref{ContractionMap} can be treated as an approximate group homomorphism. While their result holds for a general category of groups, we state it here for $ \SO(3) $:
	\begin{theorem}
		\label{NirsTheorem}
		Let $ g_{j} = (\mathbf{x}_{j}, R_{j}) \in \SE(2) $ ($ j=1,2 $). 
		Then $ \Psi_{\lambda} (g^{-1} ) = \Psi_{\lambda} (g)^{-1} $ for all $ g \in \SE(2) $ 
		and if also $ \Ltwonorm{\mathbf{x}_{1}} + \Ltwonorm{\mathbf{x}_{2}} \le 0.59 \lambda $, then for all $ g_{1}, g_{2} \in \SE(2) $
		\begin{equation}
			\label{NirsTheoremInequality}
			\Ltwonorm{\Psi_{\lambda} (g_{1} g_{2}) - \Psi_{\lambda} (g_{1}) \Psi_{\lambda} (g_{2}) }_{F}
			\le \frac{C}{\lambda^{2}}, \mbox{ as } \lambda \to \infty,
		\end{equation}
		while $ C $ is a constant independent of $ \lambda $, but dependent on $ g_{1} $ and $ g_{2} $.
		Here $ \Ltwonorm{\cdot}_{F} $ is the Frobenius norm and the codomain of $ \Psi_{\lambda} $ is explicitly the $ 3 \times 3 $ real orthogonal matrices with determinant~$ 1 $.
	\end{theorem}
	
	In the following sections, we extend this result to functions on homogeneous spaces of $ \SE(2) $ and $ \SO(3) $. In our extension, \Cref{GroupActionApproxTheorem}, the dependence on $ \lambda $ is no longer asymptotic and it is no longer required that $ \Ltwonorm{\mathbf{x}_{1}} + \Ltwonorm{\mathbf{x}_{2}} \le 0.59 \lambda $.
	
	\subsection{A mapping between functions on the plane and functions on the sphere}
	\label{ProjectionToSphereSection}
	In this section, we prove that the contraction map \eqref{ContractionMap} between $ \SE(2) $ and $ \SO(3) $ induces an invertible mapping between function spaces on homogeneous spaces of these groups. 
	A homogeneous space is a space a group acts on transitively, i.e., where every two elements in the space are related by the action of an element of the group.
	Every homogeneous space $ H $ of a some group $ G $ is $ H \cong G / G_{0} $ for some subgroup $ G_{0} \le G $ (see, e.g., \cite[Chp. II(4)]{Helgason1978}). In particular, $ \F^{2} \cong \SE(2) / \SO(2)  $ and $ S^{2} \cong \SO(3) / \SO(2) $. 
	Denote their quotient maps by $ \SubmerSE : \SE(2) \to \F^{2} $ and $ \SubmerSO : \SO(3) \to S^{2} $, respectively. 
	It is important to note that in the latter case, we identify $ \SO(2) $ with the subgroup of $ \SO(3) $ preserving the north pole $\mathbf{n} = (0, 0, 1)^{\top} $.
	Therefore, we write $ \SubmerSE(\mathbf{x}, R) = \mathbf{x} $ and $ \SubmerSO (R) = R \mathbf{n} $.
	
	Using this notation, we prove the existence of a projection of compactly supported functions on the plane onto functions on the sphere, or more informally, a projection of images from the plane onto the sphere. This projection is shown to be induced by the contraction map \eqref{ContractionMap}. Stated formally, we prove the following:
	\begin{theorem}
		\label{ProjectionTheorem}
		Fix $ \lambda > 0  $. 
		If $ f : \F^{2} \to \F $ is a smooth function compactly supported within $ B_{ \lambda \pi } = \left\{ \mathbf{x} \in \F^{2} \setsep \Ltwonorm{\mathbf{x}} < \lambda \pi \right\} $, then:
		\begin{itemize}
			\item \label{ProjectionTheoremExistence} \textsc{Existence.}
			There is a smooth function $ \kappa_{\lambda} f : S^{2} \to \F $ satisfying
			\begin{equation}
				\label{CommutativeDiagram}
				\kappa_{\lambda} f \circ \SubmerSO \circ \Psi_{\lambda} 
				= f \circ \SubmerSE ,
			\end{equation}
			or  equivalently, that the diagram in \autoref{ProofDiagram} commutes.
			
			\item \label{ProjectionTheoremUniqueness} \textsc{Uniqueness.}
			This $ \kappa_{\lambda} f $ is the only smooth function on the sphere satisfying \eqref{CommutativeDiagram}.
		\end{itemize}
	\end{theorem}
	Throughout the paper, we refer to $ \kappa_{\lambda} f $ as the projection of $ f $ onto the sphere and to $ \lambda $ as the scaling parameter of the projection.
	
	\autoref{ProofDiagram} can be taken as a visual demonstration of \Cref{ProjectionTheorem}. In it, we see two equivalent paths from $ \SE(2) $ to $ \F $. The blue path passes through the quotient map and a given function $ f : \F^{2} \to \F $. In \Cref{ProjectionTheorem}, we prove the existence and uniqueness of the green path, the one passing through $ \SO(3) $ and $ S^{2} $. Most of this latter path is already given to us. We use the properties of the contraction map $ \Psi_{\lambda} $ and the quotient map $ \SO(2) \to S^{2} $ to complete the green path.
	
	\begin{figure}
		\begin{center}
			\begin{tikzcd}
				\F^{2} \rtimes \SO(2) \arrow[d, "\SubmerSE" ,blue] \arrow[rr, "\Psi_{\lambda}",darkgreen] & & \SO(3) \arrow[d, "\SubmerSO" ,darkgreen] \\
				\F^{2}\arrow[dr, "f ",blue] & & S^{2} \arrow[dl, dashed, darkgreen, "\kappa_{\lambda} f "] \\
				& \F &  
			\end{tikzcd}
			\caption{\small\textbf{\textit{Diagram of \Cref{ProjectionTheorem}.}} There are two equivalent paths from $ \SE(2) $ to $ \F $. The blue path passes through a function on the plane. We complete the green path by proving the existence of $ \kappa_{\lambda}f $ (dashed).}
			\label{ProofDiagram}
		\end{center}
	\end{figure}
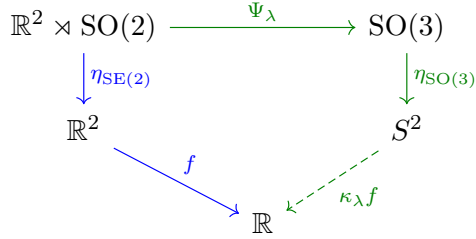
	
	\begin{remark}
		In this paper, smooth functions are $ C^{\infty} $ functions and all smooth manifolds are also $ C^{\infty} $. \Cref{ProjectionTheorem} and several other results in this paper are proved only for smooth functions. While this limits their generality, it is in fact consistent with the our primary motivation. In cryo-EM, the potential density maps are all smooth. Therefore, restricting ourselves to smooth functions,  allows us to leverage the well-developed and well-known theory of smooth manifolds, while still ensuring the results retain their potential applicability.
	\end{remark}
	
	Our proof of \Cref{ProjectionTheorem} relies on several lemmas. We state the lemmas here and leave their proofs to \Cref{ProofsAppendix}. The first lemma provides an important property of the quotient map $ \SubmerSE $ and of $ \SubmerSO \circ \Psi_{\lambda}  $. Recall that a function $ f : M \to N $ between smooth manifolds $ M $ and $ N $ is a smooth submersion if it is a smooth function with surjective differential at every point in its domain; that is, $ d_{p} f : T_{p} M \to T_{f(P)} N $ is surjective for every $ p \in M $, where $ T_{p} M $ and $ T_{f(p)} N $ are the tangent spaces of $ M $ and $ N $ at $ p $ and $ f(p) $, respectively.
	\begin{lemma}
		\label{SmoothSubmersionLemma}
		The maps
		$ \SubmerSE $ and $ \SubmerSO \circ \Psi_{\lambda} $  are smooth submersions.
	\end{lemma}
	
	Since 
	$ \SubmerSO (R) = R \mathbf{n} $, the mapping $ \SubmerSO \circ \Psi_{\lambda} $ satisfies
	\begin{equation}
		\label{EtaSOPsiLambdaFormula}
		\SubmerSO \circ \Psi_{\lambda} \left( \mathbf{x}, R \right)
		= \exp_{\SO(3)} ( \mathbf{x} ) R \mathbf{n}
		= \exp_{\SO(3)} ( \mathbf{x} ) \mathbf{n},
	\end{equation}
	where the last transition follows from our identification of $ \SO(2) \le \SO(3) $ with the subgroup of rotations preserving $ \mathbf{n} $.
	Overall, it follows that $ \SubmerSO \circ \Psi_{\lambda} (\mathbf{x}, R) $ depends only on $ \mathbf{x} $ and can be treated as a function mapping the plane onto the sphere. \autoref{MapVisualDemonstration} visually demonstrates how this mapping works, whereas \Cref{SmoothSubmersionLemmaFibers} formally characterize its fibers, the pre-image of singletons:
	\begin{lemma}
		\label{SmoothSubmersionLemmaFibers} Fix $ \lambda > 0 $. Let $ \mathbf{s} \in S^{2} \setminus \left\{ \pm \mathbf{n} \right\}$. Denote $ \widetilde{\mathbf{x}} = \left( \phi \cos \theta, \phi \sin \theta \right)^{\top} $, where $ (\theta, \phi) $ are the spherical coordinates of $ \mathbf{s} $. The fibers of $ \SubmerSO \circ \Psi_{\lambda} $ are 
		\begin{align}
			\label{Fiber1}
			\left( \SubmerSO \circ \Psi_{\lambda} \right)^{-1} \left( \mathbf{n} \right) 
			&= \left\{ (\mathbf{x}, R) \in \SE(2) \setsep \Ltwonorm{\mathbf{x}} = 2 n \lambda \pi \mbox{ for } n \in \F[N] \cup \left\{ 0 \right\},\ R \in \SO(2) \right\} \\
			\label{Fiber2}
			\left( \SubmerSO \circ \Psi_{\lambda} \right)^{-1} \left( - \mathbf{n} \right)
			&= \left\{ (\mathbf{x}, R) \in \SE(2) \setsep \Ltwonorm{\mathbf{x}} = (2n + 1) \lambda \pi \mbox{ for } n \in \F[N] \cup \left\{ 0 \right\},\ R \in \SO(2) \right\} \\
			\label{Fiber3}
			\left( \SubmerSO \circ \Psi_{\lambda} \right)^{-1} \left(\mathbf{s}\right)
			&= \left\{ (\mathbf{x}, R) \in \SE(2) \setsep \mathbf{x} = \left( 2n\pi + 1 \right) \lambda \widetilde{\mathbf{x}} \mbox{ for } n \in \F[Z],\ R \in \SO(2) \right\}.
		\end{align}
	\end{lemma}
	
	\begin{figure}
		\begin{center}
			\includegraphics[scale=0.5]{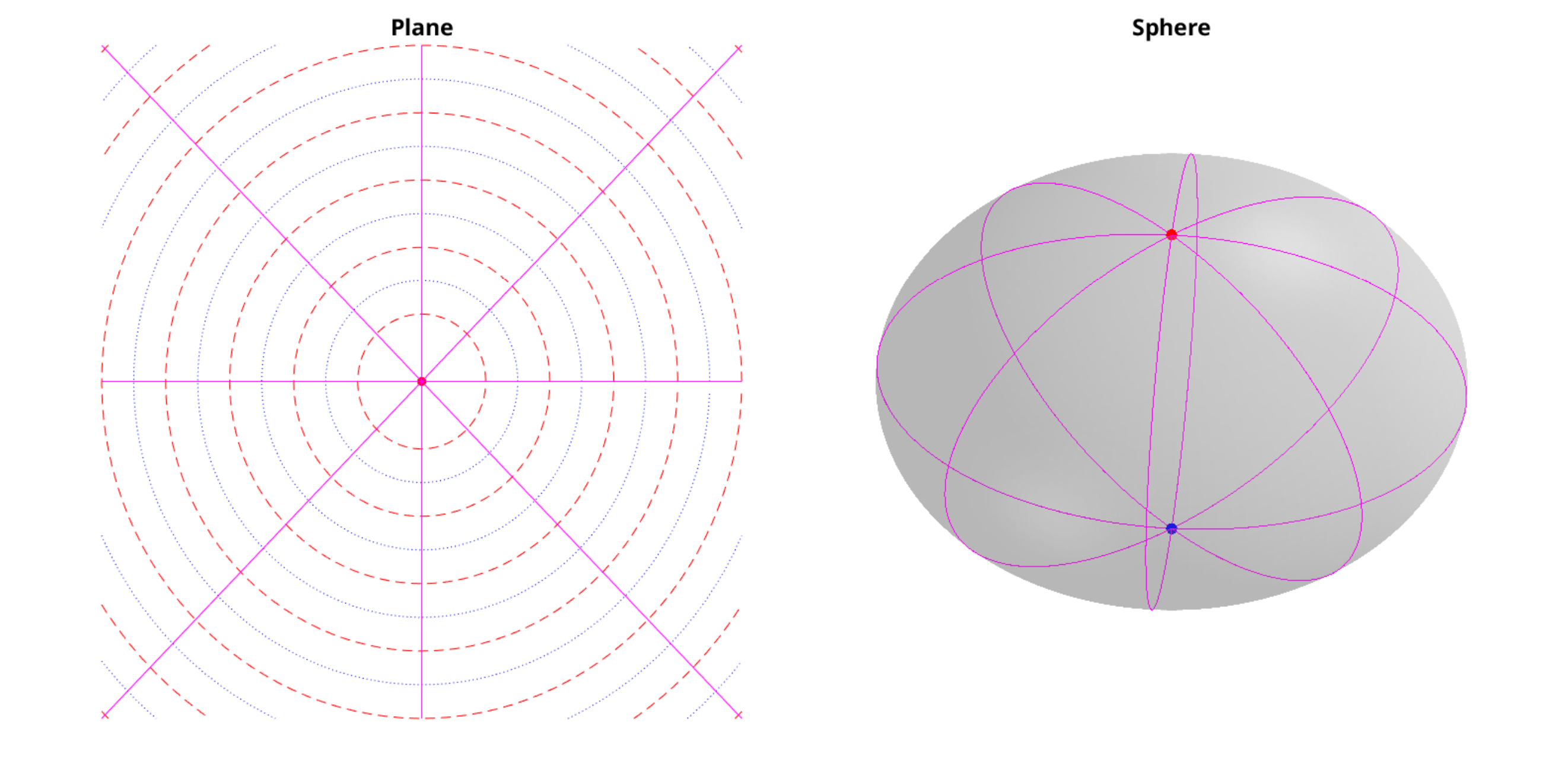}
			\caption{The mapping $ \SubmerSO \circ \Psi_{\lambda} $ can be seen as mapping lines through the origin of the plane (magenta) onto great circles through the north pole on the sphere (magenta). Consequently, the origin and the red, dashed circles are mapped to the north pole and the blue, dotted circles are mapped to the south pole.}
			\label{MapVisualDemonstration}
		\end{center}
	\end{figure}
	
	The final lemma enables us to lift a smooth, compactly supported function on the plane onto $ \SE(2) $ while respecting the fibers of $ \SubmerSO \circ \Psi_{\lambda} $:
	\begin{lemma}
		\label{ExtensionLemma}
		Let $ f : \F^{2} \to \F $ be a smooth function. If $ f $ is compactly supported within $ B_{ \lambda \pi } $, then there is a unique smooth function $ \widetilde{f} : \SE(2) \to \F $ which satisfies the following conditions:
		\begin{enumerate}[label=(\roman*)]
			
			\item \label{ExtensionLemmaCond2} $\widetilde{f}$ is constant on fibers of $ \SubmerSO \circ \Psi_{\lambda} $.
			
			\item \label{ExtensionLemmaCond3} $ f \circ \SubmerSE = \widetilde{f} $ on $ B_{ \lambda \pi } \times \SO(2) $.
		\end{enumerate}
	\end{lemma}
	
	We begin by proving the existence of the projection:
	\begin{proof}[Proof of \textsc{Existence} in \Cref{ProjectionTheorem}]
		\label{ProjectionTheoremExistenceProof}
		Let $ \widetilde{f} $ be the extension of $ f $ described in \Cref{ExtensionLemma}. 
		By \Cref{ExtensionLemma}\ref{ExtensionLemmaCond2}, it is constant on fibers of $ \SubmerSO \circ \Psi_{\lambda} $.
		Because $ \SubmerSO \circ \Psi_{\lambda} $ is a smooth submersion by \Cref{SmoothSubmersionLemma}, 
		there is $ \widetilde{F} : S^{2} \to \F $ satisfying $ \widetilde{f} = \widetilde{F} \circ \SubmerSO \circ \Psi_{\lambda} $ (see Thm. 4.30, in \cite[]{Lee2013}).
		By \Cref{ExtensionLemma}\ref{ExtensionLemmaCond3}, it follows that $ \widetilde{F} \circ \SubmerSO \circ \Psi_{\lambda} = f \circ \SubmerSE $. We therefore take $ \kappa_{\lambda} f \coloneqq \widetilde{F} $.
	\end{proof}
	
	The projection we constructed in the \hyperref[ProjectionTheoremExistenceProof]{existence proof} of \Cref{ProjectionTheorem} has an explicit, easy-to-compute formula:
	\begin{proposition}
		\label{ExplicitProjectionTheorem}
		If $ f : \F^{2} \to \F $ is a smooth compactly supported function supported within $ B_{\lambda \pi} $, then 
		\begin{equation}
			\label{ExplicitProjectionFormula}
			\kappa_{\lambda} f (\theta, \phi)
			= f(\lambda \theta \cos \phi, \lambda \theta \sin \phi).
		\end{equation}
	\end{proposition}
	Formula \eqref{ExplicitProjectionFormula} is virtually identical to the one used by Kondor \cite{Kondor2007}. While Kondor chose it based on a heuristic, we here showed that it is derived from the contraction map \eqref{ContractionMap}, a more fundamental mapping between $ \SE(2) $ and $ \SO(3) $ that relates their underlying group structure. 
	
	The proof of \Cref{ExplicitProjectionTheorem} relies on the following lemma, which provides an explicit formula for $ \SubmerSO \circ \Psi_{\lambda} $. It is proved in \Cref{ProofsAppendix}.
	\begin{lemma}
		\label{ExplicitSubmersionLemma}
		$ \SubmerSO \circ \Psi_{\lambda} (\mathbf{x}, R) = \left( \frac{\mathbf{x}^{\top}}{\Ltwonorm{\mathbf{x}}} \sin \left( \frac{\Ltwonorm{\mathbf{x}}}{\lambda} \right), \cos \left( \frac{\Ltwonorm{\mathbf{x}}}{\lambda} \right) \right)^{\top} $ for all $ (\mathbf{x}, R) \in \SE(2) $.
	\end{lemma}
	
	\begin{proof}[Proof of \Cref{ExplicitProjectionTheorem}]
		Let $ \mathbf{s} \in S^{2} \setminus \left\{ -\mathbf{n}\right\}$ with spherical coordinate representation $ (\theta, \phi) $. Let $ (\mathbf{x}, R) \in \SE(2) $ be a vector satisfying $ \SubmerSO \circ \Psi_{\lambda} (\mathbf{x}, R) = \mathbf{s} $. By \Cref{SmoothSubmersionLemmaFibers}, we can choose $ \mathbf{x} \in B_{ \lambda \pi } $. 
		By \Cref{ExplicitSubmersionLemma}, 
		\begin{equation}
			\label{lem:sPolar}
			\mathbf{s}
			= \left( \frac{\mathbf{x}^{\top}}{\Ltwonorm{\mathbf{x}}} \sin \left(\frac{\Ltwonorm{\mathbf{x}}}{\lambda}\right), \cos \left(\frac{\Ltwonorm{\mathbf{x}}}{\lambda}\right) \right)^{\top} \\
		\end{equation}
		if and only if $ \theta = \frac{\Ltwonorm{\mathbf{x}}}{\lambda} $ and $ \phi $ is the angle of $ \mathbf{x} $ in polar coordinates.
		Now, for every $ R \in \SO(2) $ we have
		\begin{align*}
			f \left( \lambda \theta \cos \phi, \lambda \theta \sin \phi \right)
			&= f \left( \Ltwonorm{\mathbf{x}} \cos \phi, \Ltwonorm{\mathbf{x}} \sin \phi \right)  
			&& \mbox{by \eqref{lem:sPolar}} \\
			&= f \left( \mathbf{x} \right) \\
			&= f \circ \SubmerSE (\mathbf{x}, R) \\
			&= \kappa_{\lambda} f \circ \SubmerSO \circ \Psi_{\lambda} (\mathbf{x}, R) 
			&& \mbox{by \Cref{ProjectionTheorem}}\\
			&= \kappa_{\lambda} f (\mathbf{s})
			= \kappa_{\lambda} f (\theta, \phi).
		\end{align*}	
	\end{proof}
	
	We use the explicit formula in \Cref{ExplicitProjectionTheorem} to prove uniqueness in \Cref{ProjectionTheorem}:
	\begin{proof}[Proof of \textsc{Uniqueness} in \Cref{ProjectionTheorem}]
		Note that the only property of $ \kappa_{\lambda} f $ we used in the proof of \Cref{ExplicitProjectionTheorem} is the fact it satisfies \eqref{CommutativeDiagram}.
		Therefore, if $ g_{1}, g_{2} : S^{2} \to \F $ satisfy \eqref{CommutativeDiagram}, they are identical on $ S^{2} \setminus \left\{-\mathbf{n} \right\} $. Finally, it is easy to see that on $ -\mathbf{n} $ both $ g_{1} $ and $ g_{2} $ must equal zero, since the support of $ f $ does not intersect the fiber of $ - \mathbf{n} $ \eqref{Fiber2}. We conclude $ g_{1} = g_{2} $, and therefore there is exactly one function on the sphere satisfying \eqref{CommutativeDiagram}.
	\end{proof}
	
	\subsection{The group action approximation theorem}
	\label{ApproximationTheoremSection}
	
	As shown above, the contraction map \eqref{ContractionMap} is a mapping between groups which induces a mapping between function spaces these groups act on. Therefore, it makes sense to try to capture the way the action of $ \SE(2) $ on its corresponding function space is approximated by the action of $ \SO(3) $ on its corresponding function space by how well the the mapping between function spaces commutes with the group action. 
	
	\Cref{NirsTheorem} of \cite{Sharon2018} does exactly that for the approximation of $ \SE(2) $ by $ \SO(3) $. We model our approach similarly to theirs and prove that
	\begin{equation*}
		\kappa_{\lambda} \left( (\mathbf{x}, R) \bullet f \right) \approx \Psi_{\lambda} ((\mathbf{x}, R)) \bullet \kappa_{\lambda} f 
		\quad\mbox{ for all $ \mathbf{x} \in \F^{2} $}.
	\end{equation*}
	
	The following theorem makes this idea precise:
	\begin{theorem}
		\label{GroupActionApproxTheorem}
		Let $ \widetilde{\lambda} \ge \lambda \ge 1 $, let $ (\mathbf{b}, R) \in \SE(2) $ and let $ f : \F^{2} \to \F $ be a smooth function compactly supported within $ B_{\lambda \pi} $. Then, 
		\begin{enumerate}[label=(\roman*)]
			\item \label{GroupActionApproxTheoremPureRot} If $ (\mathbf{b}, R) $ is a purely rotational element of $ \SE(2) $, i.e., $ \mathbf{b} = \mathbf{0} $, then 
			\begin{equation}
				\label{CummResult1}
				\kappa_{\widetilde{\lambda}} \left( (\mathbf{b}, R) \bullet f \right)
				= \Psi_{\widetilde{\lambda}} ((\mathbf{b}, R)) \bullet \kappa_{\widetilde{\lambda}} f.
			\end{equation}
			
			\item \label{GroupActionApproxTheoremNonPureRot} 
			Assume $ \left(\mathbf{b}, R\right) $ is such that 
			$ \left(\mathbf{b}, R\right) \bullet f $ is supported within $ B_{\lambda \pi} $ 
			and $ \Psi_{\widetilde{\lambda}} \left(\left(\mathbf{b}, R\right) \right) \bullet \kappa_{\widetilde{\lambda}} f $ is supported within $ \SubmerSO \circ \Psi_{\widetilde{\lambda}} \left( B_{\lambda \pi} \right)$. Then 
			\begin{equation}
				\label{CummResult2}
				\Ltwonorm{\kappa_{\widetilde{\lambda}} \left( (\mathbf{b}, R) \bullet f \right) - \Psi_{\widetilde{\lambda}} ((\mathbf{b}, R)) \bullet \kappa_{\widetilde{\lambda}} f }_{2}
				\le \frac{4 L \pi e^{\Ltwonorm{\mathbf{b}} + \lambda \pi }}{\widetilde{\lambda}^{2}} \left( 1 - \cos \left(\frac{\lambda \pi}{\widetilde{\lambda}}\right) \right)^{1/2},
			\end{equation}
			where $ L $ is the Lipschitz constant of $ \kappa_{\widetilde{\lambda}} f $ in the sense defined in \eqref{LipschitzOfProj}.
		\end{enumerate}
		
	\end{theorem}
	
	\begin{remark}
		It is worth emphasizing that the bound in \eqref{CummResult2} depends on the function being projected through $ L $, the Lipschitz constant of its projection onto the sphere.
	\end{remark}
	
	\begin{remark}
		We note an apparent difference between \Cref{GroupActionApproxTheorem} and \Cref{NirsTheorem}. 
		In \eqref{NirsTheoremInequality} of \Cref{NirsTheorem} from \cite{Sharon2018}, the bound behaves as $ O \left( \frac{1}{\lambda^{2}} \right) $ when the scaling parameter satisfies $ \lambda \to \infty $. 
		In \eqref{CummResult2} of \Cref{GroupActionApproxTheorem}, when we hold $ \lambda $ fixed and take $ \widetilde{\lambda} \to \infty $, it appears at first glances that the bound behaves as $ O \left( \frac{1}{\widetilde{\lambda}^{3}} \right) $.
		However, we expect that the the Lipschitz constant $ L $ of the projection $ \kappa_{\widetilde{\lambda}} f $ is $ O \left( \widetilde{\lambda} \right) $, and thus the bound behaves as $ O \left( \widetilde{\lambda}^{2}\right) $.
		%
		%
		To see why, note that our projection maps two fixed points on the plane to two points on the sphere, as demonstrated visually in \autoref{MapVisualDemonstration}. Therefore, if we fix two points on the plane, the distance between the value of their corresponding points on the sphere does not depend on $ \widetilde{\lambda} $. 
		However, as $ \widetilde{\lambda} $ increases, the spherical cap to which the support of $ f $ is mapped grows smaller and specifically the maximal distance between any two points in it grow smaller by $ O \left( \frac{1}{\widetilde{\lambda}} \right) $.
		Thus, the Lipschitz constant of $ \kappa_{\widetilde{\lambda}} f $ should scale asymptotically like $ L_{\lambda} O \left( \widetilde{\lambda} \right) $, where $ L_{\lambda} $ is the Lipschitz constant of $ \kappa_{\lambda} f $. Overall, we expect the entire bound to be $ O \left( \frac{1}{\widetilde{\lambda}^{2}}\right) $ as $ \widetilde{\lambda} \to \infty $ for fixed $ \lambda $, same as in \Cref{NirsTheorem}.
	\end{remark}
	
	Our proof of \Cref{GroupActionApproxTheorem} begins with a lemma providing a sense in which the matrix exponential approximately preserves the matrix group multiplication map:
	\begin{lemma}
		\label{ApproximateCumm1}
		If $ X $ and $ Y $ are real $ n \times n $ matrices and $ \lambda \ge 1 $, then
		\begin{equation}
			\label{ApproximatelyCummutative}
			\Ltwonorm{\exp \left( \frac{X}{\lambda} + \frac{Y}{\lambda} \right) - \exp \left(\frac{X}{\lambda}\right) \exp \left(\frac{Y}{\lambda}\right)}
			\le \frac{C}{\lambda^{2}},
		\end{equation}
		where 
		\begin{equation}
			\label{ApproximatelyCummutativeConstant}
			C
			\coloneqq 2 e^{\Ltwonorm{X} + \Ltwonorm{Y}} - e^{\Ltwonorm{X}} - e^{\Ltwonorm{Y}} - \Ltwonorm{X} - \Ltwonorm{Y}.
		\end{equation}
	\end{lemma}

	We particularize \Cref{ApproximateCumm1} to our contraction map between $ \SE(2) $ and $ \SO(3) $:
	\begin{corollary}
		\label{ApproximateCumm2}
		Let $ \widetilde{\lambda} \ge 1 $. If $ (\mathbf{x}_{1}, R_{1}), (\mathbf{x}_{2}, R_{1}) \in \SE(2) $, then
		\begin{equation}
			\label{ApporximateCummIneq}
			\Ltwonorm{\Psi_{\widetilde{\lambda}} \left( (\mathbf{x}_{1}, R_{1}) (\mathbf{x}_{2}, R_{2}) \right) - \Psi_{\widetilde{\lambda}} (\mathbf{x}_{1}, R_{1})  \Psi_{\widetilde{\lambda}} (\mathbf{x}_{2}, R_{2})  }
			\le \frac{C \left(\Ltwonorm{\mathbf{x}_{1}}, \Ltwonorm{\mathbf{x}_{2}}\right)}{\widetilde{\lambda}^{2}},
		\end{equation}
		where
		\begin{equation}
			\label{LemmaC}
			C\left(r_{1}, r_{2} \right)
			\coloneqq 2 e^{r_{1} + r_{2}} - e^{r_{1}} - e^{r_{2}} - r_{1} - r_{2} .
		\end{equation}
	\end{corollary}
	
	\begin{proof}
		Throughout this proof we assume that $ \exp = \exp_{\SO(3)} $.
		Since we identify $ \F^{2} \cong \Sp \left\{ S_{1}, S_{2} \right\} $ (see \Cref{ContractionMapsSection}) and the matrix exponential commutes with conjunction, for every $ R \in \SO(2) $ and  $ \mathbf{x} \in \F^{2} $ we have 
		\begin{equation}
			\label{AdjointCumm}
			\exp (R^{\top} \mathbf{x}) 
			= \bar{R}^{\top} \exp (\mathbf{x})
			\bar{R}
			\mbox{ with }
			\bar{R}
			\coloneqq \left[\begin{matrix}
				R & \mathbf{0} \\
				\mathbf{0} & 1
			\end{matrix}\right].
		\end{equation}
		
		Now, denote the left hand side of \eqref{ApporximateCummIneq} by $ M $. From \eqref{ContractionMap} and the definition of $ \bar{R} $ in \eqref{AdjointCumm}, it follows that
		\begin{align*}
			\Psi_{\widetilde{\lambda}} \left( (\mathbf{x}_{1}, R_{1}) (\mathbf{x}_{2}, R_{2})\right)
			&= \Psi_{\widetilde{\lambda}} \left( \mathbf{x}_{1} + R_{1} \mathbf{x}_{2}, R_{1} R_{2} \right)
			= \exp \left( \frac{\mathbf{x}_{1} + R_{1} \mathbf{x}_{2} }{\widetilde{\lambda}} \right) \bar{R}_{1} \bar{R}_{2} 
		\end{align*}
		and
		\begin{align*}
			\Psi_{\widetilde{\lambda}} (\mathbf{x}_{1}, R_{1}) \Psi_{\widetilde{\lambda}} (\mathbf{x}_{2}, R_{2})
			&= \exp \left( \frac{\mathbf{x}_{1}}{\widetilde{\lambda}} \right) \bar{R}_{1} \exp \left( \frac{\mathbf{x}_{2}}{\widetilde{\lambda}} \right) \bar{R}_{2} .
		\end{align*}
		Since the matrix $ 2 $-norm is rotation-invariant, it follows from the computations above that
		\begin{equation*}
			M
			= \Ltwonorm{\exp \left( \frac{\mathbf{x}_{1} + R_{1} \mathbf{x}_{2}}{\widetilde{\lambda}} \right) - \exp \left( \frac{\mathbf{x}_{1}}{\widetilde{\lambda}} \right) \bar{R}_{1} \exp \left( \frac{\mathbf{x}_{2}}{\widetilde{\lambda}} \right) \bar{R}_{1}^{\top} }.
		\end{equation*}
		From \eqref{AdjointCumm}, it now follows that
		\begin{equation*}
			M 
			= \Ltwonorm{\exp \left( \frac{\mathbf{x}_{1} + R_{1} \mathbf{x}_{2}}{\widetilde{\lambda}} \right) - \exp \left( \frac{\mathbf{x}_{1}}{\widetilde{\lambda}} \right) \exp \left( \frac{R_{1} \mathbf{x}_{2}}{\widetilde{\lambda}} \right) }.
		\end{equation*}
		Taking $ X = \mathbf{x}_{1} $, $ Y = \mathbf{x}_{2} $, it follows from \Cref{ApproximateCumm1} that
		$ M 
		\le \frac{C}{\widetilde{\lambda}^{2}} $
		with 
		\begin{align*}
			C
			&= 2 e^{\Ltwonorm{\mathbf{x}_{1}} + \Ltwonorm{R_{1} \mathbf{x}_{2}}} - e^{\Ltwonorm{\mathbf{x}_{1}}} - e^{\Ltwonorm{R_{1} \mathbf{x}_{2}}} - \Ltwonorm{\mathbf{x}_{1}} - \Ltwonorm{R_{1} \mathbf{x}_{2}} \\
			&= 2 e^{\Ltwonorm{\mathbf{x}_{1}} + \Ltwonorm{\mathbf{x}_{2}}} - e^{\Ltwonorm{\mathbf{x}_{1}}} - e^{\Ltwonorm{\mathbf{x}_{2}}} - \Ltwonorm{\mathbf{x}_{1}} - \Ltwonorm{\mathbf{x}_{2}} .
		\end{align*}
	\end{proof}
	
	\Cref{ApproximateCumm2} is an analog of \Cref{NirsTheorem} of \cite{Sharon2018}. It differs from it in three key respects. First, the requirement that the translational parts satisfy an inequality constraint is dropped. Second, the dependence on the scaling parameters $ \lambda $ and $ \widetilde{\lambda} $ is no longer asymptotic. Third, the dependence of the bound on the elements of $ \SE(2) $ is explicit.
	
	Before proving \Cref{GroupActionApproxTheorem}, we introduce another auxiliary lemmas, which we prove in Appendix A. 
	\begin{lemma}
		\label{LipschitzStuff}
		$ d_{S^{2}} (\mathbf{x}, \mathbf{y}) 
		\le \sqrt{\pi} \Ltwonorm{\mathbf{x} - \mathbf{y}} $.
		
	\end{lemma}

	We are now ready to prove our main result, \Cref{GroupActionApproxTheorem}:
	\begin{proof}[Proof of \Cref{GroupActionApproxTheorem}]
		Recall that $ \SubmerSO (R) = R \mathbf{n} $.
		Denote:
		\begin{equation}
			\label{L2NormCumm}
			\begin{aligned}
				K^{2}
				&\coloneqq \Ltwonorm{\kappa_{\widetilde{\lambda}} \left( (\mathbf{b}, R) \bullet f \right) - \Psi_{\widetilde{\lambda}} ((\mathbf{b}, R)) \bullet \kappa_{\widetilde{\lambda}} f }_{2}^{2} \\
				&= \int_{S^{2}} \left| \kappa_{\widetilde{\lambda}} \left( (\mathbf{b}, R) \bullet f \right) (\mathbf{x}) - \Psi_{\widetilde{\lambda}} ((\mathbf{b}, R)) \bullet \kappa_{\widetilde{\lambda}} f (\mathbf{x}) \right|^{2} dS^{2} (\mathbf{x}).
			\end{aligned}
		\end{equation}
		Fix $ \mathbf{x} \in S^{2} $ and let $ \mathbf{b}_{1} \in B_{\widetilde{\lambda} \pi } $ such that $ \SubmerSO \circ \Psi_{\widetilde{\lambda}} (\mathbf{b}_{1}, I) = \mathbf{x} $. We focus on two terms we subtract from one another in the integrand of \eqref{L2NormCumm}. Unraveling the definitions, it follows that the first term satisfies
		\begin{align}
			\nonumber
			\kappa_{\widetilde{\lambda}} \left( (\mathbf{b}, R) \bullet f\right) (\mathbf{x})
			&= (\mathbf{b}, R) \bullet f ( \mathbf{b}_{1}) \\
			\nonumber
			&= f \left( \left(\mathbf{b}, R\right)^{-1} \mathbf{b}_{1} \right) \\
			\nonumber
			&= f \circ \SubmerSE \left( \left(\mathbf{b}, R \right)^{-1} (\mathbf{b}_{1}, I) \right)  \\
			&= \kappa_{\widetilde{\lambda}} f \circ \SubmerSO \circ \Psi_{\widetilde{\lambda}} \left( \left( \mathbf{b}, R \right)^{-1} \left(\mathbf{b}_{1}, I \right) \right)
			\nonumber \\
			\label{Cumm1Final} 
			&= \kappa_{\widetilde{\lambda}} f \left( \Psi_{\widetilde{\lambda}} \left( \left(\mathbf{b}, R \right)^{-1} (\mathbf{b}_{1}, I) \right) \mathbf{n} \right).
		\end{align}
		Similarly, the second term satisfies
		\begin{align}
			\nonumber
			\Psi_{\widetilde{\lambda}} ((\mathbf{b}, R)) \bullet \kappa_{\widetilde{\lambda}} f (\mathbf{x})
			&= \kappa_{\widetilde{\lambda}} f( \left(\exp (\mathbf{b} ) R\right)^{\top} \exp(\mathbf{b}_{1}) \mathbf{n} ) \\
			&= \kappa_{\widetilde{\lambda}} f( \Psi_{\widetilde{\lambda}} \left( \left(\mathbf{b}, R\right) \right)^{-1} \exp(\mathbf{b}_{1}) \mathbf{n} ) \nonumber \\
			&= \kappa_{\widetilde{\lambda}} f \left( \Psi_{\widetilde{\lambda}} \left( \left(\mathbf{b}, R\right)^{-1} \right) \Psi_{\widetilde{\lambda}}  ((\mathbf{b}_{1}, I) \mathbf{n}) \right),
			\label{Cumm2Final}
		\end{align}
		where the last transition follows from the first half of \Cref{NirsTheorem}, which does not require any restrictions on $ \lambda $.
		Now, note that $ (\mathbf{b}, R) = (0, R) (\mathbf{b}, I) $ and that $ \Psi_{\widetilde{\lambda}} (\mathbf{b}, R) = \Psi_{\widetilde{\lambda}} (0, R) \Psi_{\widetilde{\lambda}} (\mathbf{b}, I) $. Therefore, substituting $ \mathbf{b} = \mathbf{0} $ into \eqref{Cumm1Final} and $ \eqref{Cumm2Final} $ shows that the left-hand side of both equations are equal if $ (\mathbf{b}, R) $ is a purely rotational element of $ \SE(2) $. Thus, the integrand in \eqref{L2NormCumm} is zero on $ S^{2} $, which proves \eqref{CummResult1}.
		
		In order to prove \eqref{CummResult2}, let $ \alpha = \Psi_{\widetilde{\lambda}} \left( \left( \mathbf{b}, R \right)^{-1} (\mathbf{b}_{1}, I) \right) $ and $ \beta = \Psi_{\widetilde{\lambda}} \left( \left( \mathbf{b}, R \right)^{-1} \right) \Psi_{\widetilde{\lambda}} (\mathbf{b}_{1}, I) $. 
		Since $ \kappa_{\widetilde{\lambda}} f $ is a smooth function with a compact domain $ S^{2} $, it is Lipschitz in the sense that there is $ L \ge 0 $ satisfying
		\begin{equation}
			\label{LipschitzOfProj}
			\left| \kappa_{\widetilde{\lambda}} f(\mathbf{x}) - \kappa_{\widetilde{\lambda}} f(\mathbf{y}) \right| \le L d_{S^{2}} ( \mathbf{x}, \mathbf{y} )
			\enskip \mbox{for all } \mathbf{x}, \mathbf{y} \in S^{2},
		\end{equation}
		where $ d_{S^{2}} (\mathbf{x}, \mathbf{y}) = \arccos(\mathbf{y}^{\top} \mathbf{x}) $ is the standard metric on the sphere (great-circle distance). 
		Therefore, 
		\begin{equation}
			\label{MainResultOfProof}
			\begin{aligned}
				| \kappa_{\widetilde{\lambda}} \left( (\mathbf{b}, R) \bullet f \right) (\mathbf{x}) &- \Psi_{\widetilde{\lambda}} ((\mathbf{b}, R)) \bullet \kappa_{\widetilde{\lambda}} f (\mathbf{x}) |\\
				&= \left| \kappa_{\widetilde{\lambda}} f (\alpha \mathbf{n}) - \kappa_{\widetilde{\lambda}} f (\beta \mathbf{n}) \right|, && \mbox{By \eqref{Cumm1Final} and \eqref{Cumm2Final}} \\
				&\le L d_{S^{2}} (\alpha \mathbf{n}, \beta \mathbf{n}), && \mbox{By \eqref{LipschitzOfProj}} \\
				&\le L \sqrt{\pi} \Ltwonorm{\alpha \mathbf{n} - \beta \mathbf{n}}, && \mbox{By \Cref{LipschitzStuff}} \\
				&\le L \sqrt{\pi} \Ltwonorm{\alpha - \beta}, && \mbox{Since $ \Ltwonorm{\mathbf{n}} = 1 $} \\
				&\le \frac{2 L \sqrt{\pi} C \left( \Ltwonorm{\mathbf{b}}, \Ltwonorm{\mathbf{b}_{1}} \right)}{\widetilde{\lambda}^{2}}  && \mbox{By \Cref{ApproximateCumm2}},
			\end{aligned}
		\end{equation}
		where $ C $ is defined in \eqref{LemmaC}.
		
		From \Cref{ExplicitSubmersionLemma}, it is easy to show that if we treat $ \SubmerSO \circ \Psi_{\widetilde{\lambda}}  $ as map from $ \F^{2} $ to the sphere, it can be represented in polar coordinates as
		\begin{equation*}
			\SubmerSO \circ \Psi_{\widetilde{\lambda}} \left( r, \phi \right)
			= \left( \cos \phi \sin \left(\frac{r}{\widetilde{\lambda}}\right), \sin \phi \sin \left( \frac{r}{\widetilde{\lambda}} \right), \cos \left(\frac{r}{\widetilde{\lambda}}\right) \right)^{\top} .
		\end{equation*}
		Furthermore, it is obvious that it is a smooth parameterization of $ S^{2} \setminus \left\{-\mathbf{n}\right\} $ with a Gram determinant
		\begin{equation*}
			\det G 
			= \frac{1}{\widetilde{\lambda}^{2}} \sin^{2} \left( \frac{r}{\widetilde{\lambda}}\right) .
		\end{equation*}
		Finally, since by hypothesis $ \left(\mathbf{b}, R\right) \bullet f $ is supported within $ B_{\lambda \pi} $, both $ \Psi_{\widetilde{\lambda}} \left(\left(\mathbf{b}, R\right) \right) \bullet \kappa_{\widetilde{\lambda}} f $ and $ \kappa_{\widetilde{\lambda}} \left( \left(\mathbf{b}, R\right) \bullet f \right) $ are supported within $ \SubmerSO \circ \Psi_{\widetilde{\lambda}} \left( B_{\lambda \pi} \right)$. 
		Overall, we can write \eqref{L2NormCumm} as
		\begin{equation*}
			K^{2} 
			= \frac{1}{\widetilde{\lambda}} \int_{0}^{\lambda \pi } \int_{0}^{2 \pi} \left| \kappa_{\widetilde{\lambda}} \left( (\mathbf{b}, R) \bullet f \right) (\mathbf{y}) - \Psi_{\widetilde{\lambda}} ((\mathbf{b}, R)) \bullet \kappa_{\widetilde{\lambda}} f (\mathbf{y}) \right|^{2} \sin \left( \frac{r}{\widetilde{\lambda}}\right)  d\phi d r, 
		\end{equation*}
		where $ \mathbf{y} = \SubmerSO \circ \Psi_{\widetilde{\lambda}} \left( \phi, r\right)  $. 
		From \eqref{MainResultOfProof} we obtain
		\begin{align*}
			K^{2} 
			&\le \left( \frac{2L \sqrt{\pi} }{\widetilde{\lambda}^{2}}\right)^{2}\frac{1}{\widetilde{\lambda}} \int_{0}^{\lambda \pi} \int_{0}^{2 \pi} C \left(\Ltwonorm{\mathbf{b}}, r\right)^{2} \sin\left(\frac{r}{\widetilde{\lambda}}\right) d\phi d r \\
			&= \left( \frac{2L \sqrt{\pi} }{\widetilde{\lambda}^{2}}\right)^{2}\frac{2 \pi }{\widetilde{\lambda}} \int_{0}^{\lambda \pi} C \left(\Ltwonorm{\mathbf{b}}, r\right)^{2} \sin\left(\frac{r}{\widetilde{\lambda}}\right) d r.
		\end{align*}
		It is clear that $ C $ in \eqref{LemmaC} satisfies 
		\begin{equation*}
			0 
			\le C \left(\Ltwonorm{\mathbf{b}} , r\right) 
			\le 2 e^{\Ltwonorm{\mathbf{b}} + r } 
			\le 2 e^{\Ltwonorm{\mathbf{b}} + \lambda \pi } , 
		\end{equation*}
		for $ 0\le  r \le \lambda \pi  $.
		It follows that
		\begin{align*}
			K^{2} 
			&\le \left( \frac{2L \sqrt{\pi} }{\widetilde{\lambda}^{2}}\right)^{2}\frac{2 \pi }{\widetilde{\lambda}} \cdot 2 e^{2 \left(\Ltwonorm{\mathbf{b}} + \lambda \pi\right) }  \cdot   \int_{0}^{\lambda \pi} \sin\left(\frac{r}{\widetilde{\lambda}}\right) d r \\
			&= \left( \frac{2L \sqrt{\pi} }{\widetilde{\lambda}^{2}}\right)^{2} \cdot 2 \pi \cdot 2 e^{2 \left(\Ltwonorm{\mathbf{b}} + \lambda \pi\right) }  \cdot \left( - \cos \left(\frac{r}{\widetilde{\lambda}}\right) \right) \bigg|_{0}^{\lambda \pi} \\
			&= \left( \frac{2L \sqrt{\pi} }{\widetilde{\lambda}^{2}}\right)^{2} \cdot 2 \pi \cdot 2 e^{2 \left(\Ltwonorm{\mathbf{b}} + \lambda \pi\right) }  \cdot \left( 1 - \cos \left(\frac{\lambda \pi}{\widetilde{\lambda}}\right) \right).
		\end{align*}
		Taking the square root of both sides, we obtain
		\begin{equation*}
			K 
			\le \frac{4 L \pi e^{\Ltwonorm{\mathbf{b}} + \lambda \pi }}{\widetilde{\lambda}^{2}} \left( 1 - \cos \left(\frac{\lambda \pi}{\widetilde{\lambda}}\right) \right)^{1/2}.
		\end{equation*}
	\end{proof}

	\subsection{The spherical bispectrum as an $ \SO(3) $-invariant representation of functions on the sphere}
	\label{SphericalBispectrumSection}
	\Cref{GroupActionApproxTheorem} suggests an approximately $ \SE(2) $-invariant representation for functions on the plane: project the function onto the sphere and then compute an orbit-characterizing $ \SO(3) $-invariant representation of the projected function. To that end, we present the spherical bispectrum, an orbit-characterizing $ \SO(3) $-invariant representation of functions on the sphere.
	The theory of the spherical bispectrum was largely developed by \cite{Kakarala1992,Kakarala1993,Kakarala2010}. These works heavily relied on tools from harmonic analysis over groups. 
	Here we merely draw on their results and do not revisit their theoretical underpinnings. Additionally, when needed, we drew material on spherical harmonics from \cite[Chp. 2]{Atkinson2012}.
	
	Let $ \LtwoStwo $ be the Hilbert space of all square-integrable complex-valued spherical functions equipped with the inner product 
	\begin{equation}
		\label{InnerProductSphere}
		\Eucprod{f}{g}_{S^{2}} 
		= \int_{S^{2}} f(\mathbf{x}) g(\mathbf{x})^{*} dS^{2} (\mathbf{x}).
	\end{equation}
	It induces a norm $ \Ltwonorm{f}_{S^{2}} = \sqrt{\Eucprod{f}{f}} $.
	Here $ dS^{2}  $ is the standard volume element on $ S^{2} $. 
	$ \LtwoStwo $ is the direct sum of the of spaces of spherical harmonics of arbitrary degree. To make this statement more precise, let $ \mathcal{H}_{\ell} $ be the space of spherical harmonics of degree $ \ell \in \F[Z] $; that is, 
	\begin{equation}
		\label{SphericalHarmonicsSpace}
		\mathcal{H}_{\ell} = \Sp \left\{ Y_{\ell,m} :S^{2} \to \F[C] \setsep m = -\ell, -\ell+1, \dots, \ell \right\},
	\end{equation}
	where the spherical harmonic function of degree $ \ell $ and order $ m $ are given by
	\begin{equation}
		\label{SphericalHarmonics}
		Y_{\ell, m} (\theta, \phi)
		= \sqrt{\frac{2 \ell + 1}{4 \pi} \frac{(\ell-m)!}{(\ell+m)!} } P_{\ell,m} (\cos \theta) e^{i m \phi}.
	\end{equation}
	Here, $ P_{\ell,m} (x) $ is the associated Legendre function of degree $ \ell $ and order $ m $. 
	The collection of all spherical harmonic functions is an orthonormal basis of $ \LtwoStwo $. Also, $ R \bullet \mathcal{H}_{\ell} \subseteq \mathcal{H}_{\ell} $ for all $ R \in \SO(3) $ and $ \ell $. 
	
	Fix $ f \in \LtwoStwo $. Since spherical harmonics \eqref{SphericalHarmonics} are orthonormal with respect to inner product \eqref{InnerProductSphere}, we can expand $ f $ in terms of \eqref{SphericalHarmonics} by taking as coefficients $ f_{\ell,m} = \Eucprod{f}{Y_{\ell,m}}_{S^{2}}  $. We say $ f $ has bandlimit $ L $ if $ f_{\ell,m} = 0 $ for all $ \ell > L $. Combining the results of \cite{Kakarala1992} and \cite{Kakarala2010}, the spherical bispectrum of $ f $ is the set of all numbers 
	\begin{equation}
		\label{SphericalBispectrum}
		b_{f} [\ell_{1}, \ell_{2}, \ell]
		= \sum_{m=-\ell}^{\ell}  \sum_{m_{1} = -\ell_{1}}^{\ell_{1}} C_{\ell_{1},m_{1}, \ell_{2}, m-m_{1} }^{\ell, m} f_{\ell,m} f_{\ell_{1},m_{1}}^{*} f_{\ell_{2},m-m_{1}}^{*}, 
	\end{equation}
	where $ \ell_{1} $ and $ \ell_{2} $ are non-negative integers and $ \left| \ell_{1}- \ell_{2} \right| \le \ell \le  \ell_{1} + \ell_{2} $. The numbers $ C_{\ell_{1},m_{1}, \ell_{2}, m_{2} }^{\ell, m} $ are the Clebsch-Gordan coefficients. These are real numbers related to the underlying group structure of $ \SO(3) $. We detail a numerical algorithm for their calculation in \Cref{BispectrumComputationSection}.
	
	The spherical bispectrum \eqref{SphericalBispectrum} is intimately related to the triple correlation function of $ f $, which is defined by:
	\begin{equation*}
		T_{f} (R_{1}, R_{2})
		= \int_{\SO(2)} f(\mathbf{x})^{*} f (R_{1} \mathbf{x}) f(R_{2} \mathbf{x}) d \mu (\mathbf{x}),
	\end{equation*}
	where $ d \mu $ is the Haar measure on $ \SO(2) $. 
	Kakarala and Mao \cite[Theorem 4.1(a)]{Kakarala2010} proved that calculating \eqref{SphericalBispectrum} is equivalent to calculating the Fourier transform on the group $ \SO(3) \times \SO(3) $ of the triple correlation of $ f $. 
	Furthermore, the Fourier transform of a the triple correlation of a bandlimited real-valued spherical function was shown in \cite{Kakarala1993} to uniquely determine this function up to a rotation, provided it is non-zero; that is, if for every $ 0 \le \ell \le L $ there is $ - \ell \le m \le \ell $ such that $ f_{\ell,m } \ne 0 $, $ b_{g} = b_{f} $ if and only if $ g = R \bullet f $ for some $ R \in \SO(3) $. 
	Thus, the spherical bispectrum \eqref{SphericalBispectrum} determines the orbit of a real-valued a function of bandlimit $ L $ with non-zero projection onto $ \mathcal{H}_{\ell} $ for all $ 0 \le \ell \le L $. 
	
	The specific form of the bispectrum we use in \eqref{SphericalBispectrum} is owed to Kondor \cite{Kondor2007}, but is easily derivable from \cite[eq. (24)]{Kakarala2010}. 
	Furthermore, Kakarala and Mao \cite[Theorem 4.1(b)]{Kakarala2010} showed that for every $ \left( \ell_{1}, \ell_{2}, \ell \right) $ we have $ b_{f} \left[\ell_{1}, \ell_{2} ,\ell\right]  = \zeta b_{f} \left[\ell_{2}, \ell_{1} ,\ell\right] $ for some $ \zeta \in \F[C] $ independent of $ f $. Thus, it is sufficient to compute \eqref{SphericalBispectrum} for $ \ell_{2} \le \ell_{1} $.
	Also, it is obvious that if $ \max\left\{\ell_{1}, \ell_{2}, \ell \right\} > L$, then \eqref{SphericalBispectrum} is zero for functions of bandlimit $ L $. 
	
	Overall, given a real-valued function $ f \in \LtwoStwo $ of bandlimit $ L $ which has a non-zero projection over $ \mathcal{H}_{\ell} $ for all $ 0\le \ell \le L $, we can represent it up to rotation by its spherical bispectrum \eqref{SphericalBispectrum} over all triplets $ (\ell_{1}, \ell_{2}, \ell) $ satisfying
	\begin{equation}
		\label{SphericalBispectrumIndices}
		0 \le \ell_{1} \le L, \enskip
		0 \le \ell_{2} \le \ell_{1}, \enskip
		\ell_{1} - \ell_{2} \le \ell \le \min\left\{ L, \ell_{1}+ \ell_{2}\right\}.
	\end{equation}
	This representation uniquely determines $ f $ up to rotation; that is, it characterizes the orbit \\
	$ \SO(3) \bullet f $.
	
	\section{Compactifiaction of functions on the plane: computational aspects}
	\label{CompactifiactionComputationSection}
	
	We developed the theory in \Cref{CompactifiactionTheorySection} for functions on the plane. Here we detail several computational aspects that arise when working with discrete samples of such functions.
	In \Cref{ProjectionSection}, we explain how we use discrete samples of a function on the plane to estimate the spherical harmonics coefficients of the projection of this function onto the sphere. 
	In \Cref{BispectrumComputationSection}, we explain how the spherical bispectrum \eqref{SphericalBispectrum} is evaluated for bandlimited functions on the sphere. Most importantly, we present a method, introduced by \cite{Straub2014}, for the evaluation of the Clebsch-Gordan coefficients.
	We demonstrate the accuracy of this method. We show numerically that the spherical bispectrum also echoes the results of \Cref{GroupActionApproxTheorem}. In particular, for sufficiently small translation sizes, the spherical bispectrum of an image and its translation are approximately the same.

	\subsection{Projecting an image onto the sphere}
	\label{ProjectionSection}
	Let $ f : \F^{2} \to \F $ be a function compactly supported within $ \left[-\zeta,\zeta\right]^{2} $, where $ \zeta = \cos \frac{\pi}{4} = \sin \frac{\pi}{4} $. This implies it is compactly supported within $ B_{\pi} $.
	Let $ G = \left\{(x_{i}, y_{j})\right\}_{i=0,j=0}^{n-1} $ ($ n>1 $) be a regular equidistant grid on $ \left[-\zeta, \zeta \right]^{2} $:
	\begin{equation}
		\label{GridDefinition}
		x_{i} 
		= -\zeta + \frac{2\zeta }{n-1}i, \
		i = 0, 1,\dots,n-1,
		\enskip
		\mbox{ and }
		\enskip
		y_{j} = -\zeta + \frac{2\zeta }{n-1}j, \
		j = 0, 1,\dots,n-1.
	\end{equation}
	Let $ I  $ be an a real $ n \times n $ matrix with $ I_{i,j} = f(x_{i}, y_{j}) $.
	Given a fixed $ \lambda \ge 1 $, our objective is to use the discrete data available to us to approximate the spherical harmonics coefficients of $ \kappa_{\lambda} f $ under the assumption that it has a bandlimit $ L $.
	
	Since the spherical harmonics \eqref{SphericalHarmonics} are orthonormal, our task amounts to approximating the integral
	\begin{equation}
		\label{IntegralToApproximate}
		\Eucprod{\kappa_{\lambda} f }{Y_{\ell,m}}
		= \int_{S^{2}} \kappa_{\lambda} f(\mathbf{x}) Y_{\ell,m} (\mathbf{x})^{*} d S^{2} (\mathbf{x}).
	\end{equation}
	Integrals of this form can be approximated well using spherical designs. A spherical $ t $-design is set of spherical points $ \mathcal{S}_{t} = \left\{\mathbf{x}_{1}, \mathbf{x}_{2}, \dots, \mathbf{x}_{N_{t}} \right\}  $ satisfying 
	\begin{equation*}
		\int_{S^{2}} Y_{\ell,m} (\mathbf{x}) d S^{2} (\mathbf{x})
		= \frac{1 }{N_{t}} \sum_{n=1}^{N_{t}} Y_{\ell, m} (\mathbf{x}_{n}),
		\quad 
		0 \le \ell \le t.
	\end{equation*}
	Since we assume $ \kappa_{\lambda} f $ has a bandlimit $ L $ and since $ Y_{\ell_{1}, m_{1}} \cdot Y_{\ell_{2}, m_{2}}^{*} \in \mathcal{H}_{\ell_{1} + \ell_{2}} $ \cite{Atkinson2012}, we approximate \eqref{IntegralToApproximate} using a spherical $ 2L $-design, thus:
	\begin{equation}
		\label{IntegralApproximation}
		f_{\ell,m}
		\approx \frac{1 }{N} \sum_{n=1}^{N} \kappa_{\lambda} f(\theta_{n}, \phi_{n} ) Y_{\ell, m} (\theta_{n}, \phi_{n} )^{*},
	\end{equation}
	where $ (\theta_{n}, \phi_{n} )  $ is $ \mathbf{x}_{n} $ in spherical coordinates.
	This can be written in matrix form. Let $ \mathbf{Y} = \mathbf{Y}(L, \mathcal{S}_{2L}) $ be  the complex matrix with rows index by $ 1,\dots,N_{t} $  and columns indexed by $ (\ell,m) $ defined by 
	\begin{equation}
		\label{SphericalHarmonicsMatrix}
		\mathbf{Y}_{n, (\ell,m)}
		= Y_{\ell,m} (\theta_{n}, \phi_{n} ).
	\end{equation}
	Our estimate of the spherical harmonics coefficients of $ f $ is then $ \frac{1}{N_{t}} \mathbf{Y}^{\dag} \mathbf{v} $, where $ \mathbf{v} $ are the values of $ \kappa_{\lambda} f $ on the spherical design $ \mathcal{S}_{2L} $.
	
	\begin{remark}
		Approximating \eqref{IntegralToApproximate} by evaluating \eqref{IntegralApproximation} on a  spherical $ 2L $-design might be numerically unstable, especially for large bandlimits $ L $. In order to deal with that, using a $ K $-design with $ K > 2L $ is advisable, at least when $ L $ is large. 
	\end{remark}
	
	In order to calculate \eqref{IntegralApproximation}, one needs a spherical design and the value of $ \kappa_{\lambda} f $ on it. Tables of spherical designs were obtained from \cite{Womersley2017, Womersley2017Web}. 
	The values of $ \kappa_{\lambda} f $ on these points are estimated by employing \eqref{ExplicitProjectionFormula} and interpolation. At the heart of \eqref{ExplicitProjectionFormula} is a mapping $ (\theta, \phi)  \mapsto (\lambda \theta \cos \phi, \lambda \theta \sin \phi) $ from the sphere to the plane. This mapping can easily be shown to be invertible as mapping from $ S^{2} \setminus \left\{ (0,0,-1)^{\top}\right\} $ onto $ B_{\lambda \pi} $. We map the spherical design points onto the plane and isolate those that fall inside $ \left[-\zeta,\zeta\right]^{2} $. We then interpolate the known values of $ f $ on the grid $ G $ onto these spherical design points. By \eqref{ExplicitProjectionFormula}, the values of $ f $ on any other spherical design point is necessarily zero.
	
	Three points are important to note here. 
	First, relying on interpolation complicates the handling of noisy data, since it introduces correlations in the noise. We deal with that in detail in \Cref{DenoisingSection}. 
	Second, evaluating spherical harmonics is computationally costly. Thus, when calculating $ \mathbf{Y} $ in \eqref{SphericalHarmonicsMatrix}, it is advisable to exploit the fact the values of $ \kappa_{\lambda} f $ are assumed to be zero on spherical design points that are projected to outside of $ \left[-\zeta,\zeta\right]^{2} $. We therefore remove said points from the spherical design before calculating \eqref{SphericalHarmonicsMatrix}. For brevity, we will still refer to this subset of the original spherical design as $ \mathcal{S}_{2L} $. 
	Third, if the image includes additive noise, one would want to consider how the projection affects it.
	In \Cref{NoiseStatisticsAppendix}, we discuss that for the noise model we use throughout the paper. In particular, we empirically demonstrate that additive white noise in the image approximately results in an additive white noise in the estimated spherical harmonics coefficients.
	
	\subsection{Calculating the spherical bispectrum}
	\label{BispectrumComputationSection}
	
	The procedure described in \Cref{ProjectionSection} uses discrete samples of a function on the plane to approximate the spherical harmonics coefficients up to bandlimit $ L $ of its projection onto the sphere. 
	Our objective now is to use the spherical harmonics coefficients to calculate its spherical bispectrum, its $ SO(3) $-invariant representation that we presented in \Cref{SphericalBispectrumSection}. 
	
	Let $ f $ be a bandlimited function on the sphere with spherical harmonics coefficients $ \left\{f_{\ell,m}\right\}_{\ell=0}^{L} $ up to bandlimit $ L $. The bispectrum of $ f $ is calculated by evaluating \eqref{SphericalBispectrum}. Since the Clebsch-Gordan coefficients $ C_{\ell_{1}, m_{1}, \ell_{2}, m_{2} }^{\ell,m} $ are non-zero only when $ m_{1} + m_{2} = m $ (\cite[Chp. 8, eq. (2)]{Varshalovich1988}), it follows that $ C_{\ell_{1}, m_{1}, \ell_{2}, m-m_{1}}^{\ell,m} $ is non-zero only when $ \underline{m_{1}} \le m_{1} \le \overline{m_{1} }  $ for $ \underline{m_{1}} = \max\left\{ -\ell_{1}, m-\ell_{2} \right\}  $ and $ \overline{m_{1}} = \min\left\{ \ell_{1}, m+\ell_{2} \right\} $. Thus, \eqref{SphericalBispectrum} becomes
	\begin{equation}
		\label{SphericalBispectrumSimplified}
		b_{f} [\ell_{1}, \ell_{2}, \ell]
		= \sum_{m=-\ell}^{\ell} f_{\ell,m} \sum_{m_{1} = \underline{m_{1}}}^{\overline{m_{1}}} C_{\ell_{1}, m_{1}, \ell_{2}, m-m_{1}}^{\ell,m} f_{\ell_{1}, m_{1}}^{*} f_{\ell_{2}, m-m_{1}}^{*}. 
	\end{equation}
	This expression is evaluated for all triplets $ (\ell_{1}, \ell_{2}, \ell) $ in \eqref{SphericalBispectrumIndices} and the resulting values are saved in a linear array in lexicographical ordering of the triplets. We refer to this array as the bispectrum vector and denote it by $ \mathbf{b}_{f} $. 
	
	In order to evaluate \eqref{SphericalBispectrumSimplified}, the Clebsch-Gordan coefficients need to be computed. We follow the method elaborated in \cite{Straub2014} and compute all relevant Clebsch-Gordan coefficients in advance. For a fixed $ (\ell_{1}, \ell_{2}, \ell) $ triplet and a fixed $ -\ell \le m \le \ell $, \cite{Straub2014} showed that $ \mathbf{c} = \left(C_{\ell_{1}, m_{1}, \ell_{2}, m - m_{1}}^{\ell, m} \right)_{m_{1} = \underline{m_{1}}}^{\overline{m_{1}}} $ is a solution of the linear equation  $ \mathbf{C} \mathbf{c} = \mathbf{0} $. Here, $ \mathbf{C} $ is the tridiagonal $ (n-1) \times n $ real matrix, where $ n = \overline{m_{1}} - \underline{m_{1}} + 1 $  and its diagonals are
	\begin{align*}
		\mathbf{C}_{i, i}
		&= \ell_{1} (\ell_{1} + 1) 
		+ \ell_{2} (\ell_{2}+ 1)
		+ 2 m_{1, i} m_{2, i} 
		- \ell (\ell+1) \\
		\mathbf{C}_{i,i+1}
		&= \mathbf{C}_{i+1,i}
		= \sqrt{\ell_{1} (\ell_{1} +1) - m_{1,i} m_{1, i+1}} \sqrt{\ell_{2} (\ell_{2} +1) - m_{2,i} m_{2, i+1}} ,
	\end{align*}
	where
	\begin{equation*}
		m_{1,i} 
		= \underline{m_{1}} + i - 1
		\enskip\mbox{and}\enskip 
		m_{2, i} = m - \underline{m_{1}} - i + 1.
	\end{equation*}
	This equation does not have a unique solution. However, the Clebsch-Gordan coefficients must have norm of $ 1 $ and by convention satisfy $ C_{\ell_{1}, \overline{m_{1}}, \ell_{2}, m - \overline{m_{1}}}^{\ell, m} > 0 $ (see \cite[3]{Straub2014}). We therefore find a solution for this linear equation by choosing $ C_{\ell_{1}, \overline{m_{1}}, \ell_{2}, m - \overline{m_{1}}}^{\ell, m} = \frac{1}{\overline{m_{1}} - \underline{m_{1}} + 1}$ and determining the remaining coordinates in the solution vector by backward substitution. The resulting vector is then normalized to be a unit vector.
	
	We tested how well the result of \Cref{GroupActionApproxTheorem} are echoed in the invariance of the bispectrum. 
	To that effect, we generated a random image using the following procedure. First, we generated a random set of spherical harmonics coefficients of a real-valued function on the sphere $ f $ of bandlimit $ 14 $.
	The coefficients were sampled uniformly from the set of coefficients satisfying 
	\begin{equation}
		\label{RealSHCSphericalDomain}
		\sum_{m=-\ell}^{\ell} \left|f_{\ell,m}\right|^{2} 
		= 1 
	\end{equation}
	for all $ 0 \le \ell \le 14 $. 
	Given a square-integrable function $ f : S^{2} \to \F $, its spherical harmonics coefficients are
	\begin{equation*}
		f_{\ell,m}
		= \Eucprod{f}{Y_{\ell,m}}
		= \int_{S^{2}} f(\mathbf{x}) Y_{\ell,m} (\mathbf{x})^{*} d S^{2} (\mathbf{x}).
	\end{equation*}
	Since $ Y_{\ell,m} (\mathbf{x})^{*} = (-1)^{m} Y_{\ell,-m}(\mathbf{x}) $ \cite{Varshalovich1988} and $ f $ is real-valued, it follows that
	\begin{equation}
		\label{RealSHCSymmetry}
		f_{\ell, m}^{*}
		= (-1)^{m} f_{\ell, -m}.
	\end{equation}
	Thus, the coefficients were sampled uniformly from those satisfying both \eqref{RealSHCSphericalDomain} and \eqref{RealSHCSymmetry}.
	The resulting function was back-projected onto a $ 101\times 101 $ regular equidistant grid in $ [-\zeta,\zeta]^{2} $ by evaluating \eqref{ExplicitProjectionFormula} on all grid points for scaling parameter $ \lambda = 1 $. 
	In the resulting images, the first and last $ 20 $ rows and columns of pixels were set to zero and then smoothed using a Gaussian filter. 
	This was done in order to create an image that has approximately zero in its margin, with enough room to apply translations to.
	The resulting image was projected onto the sphere by estimating its spherical harmonics coefficients up to bandlimit $ 16 $ with $ \lambda = 1 $. The final image is produced by back-projecting this spherical function onto the same grid with $ \lambda  = 1 $.
	Hereafter, we refer to these images as random images.
	\autoref{RandomImageExample} shows eight examples of random images generated using this procedure.
	
	\begin{figure}[t]
		\begin{center}
			\begin{subfigure}{\textwidth}
				\begin{center}
					\includegraphics[scale=0.6]{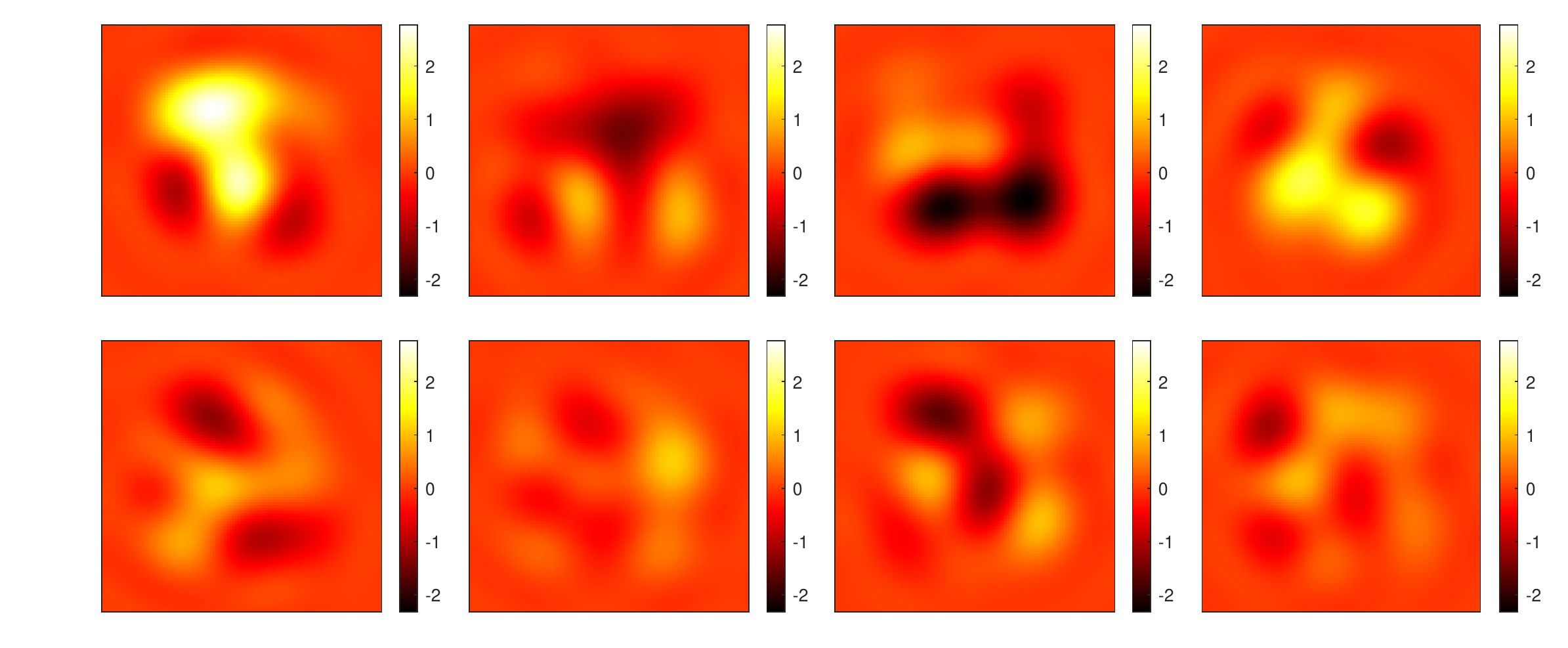}
					
					\caption{Random images from \Cref{BispectrumComputationSection}}
					\label{RandomImageExample}
				\end{center}
			\end{subfigure}
			
			\begin{subfigure}{\textwidth}
				\begin{center}
					\includegraphics[scale=0.6]{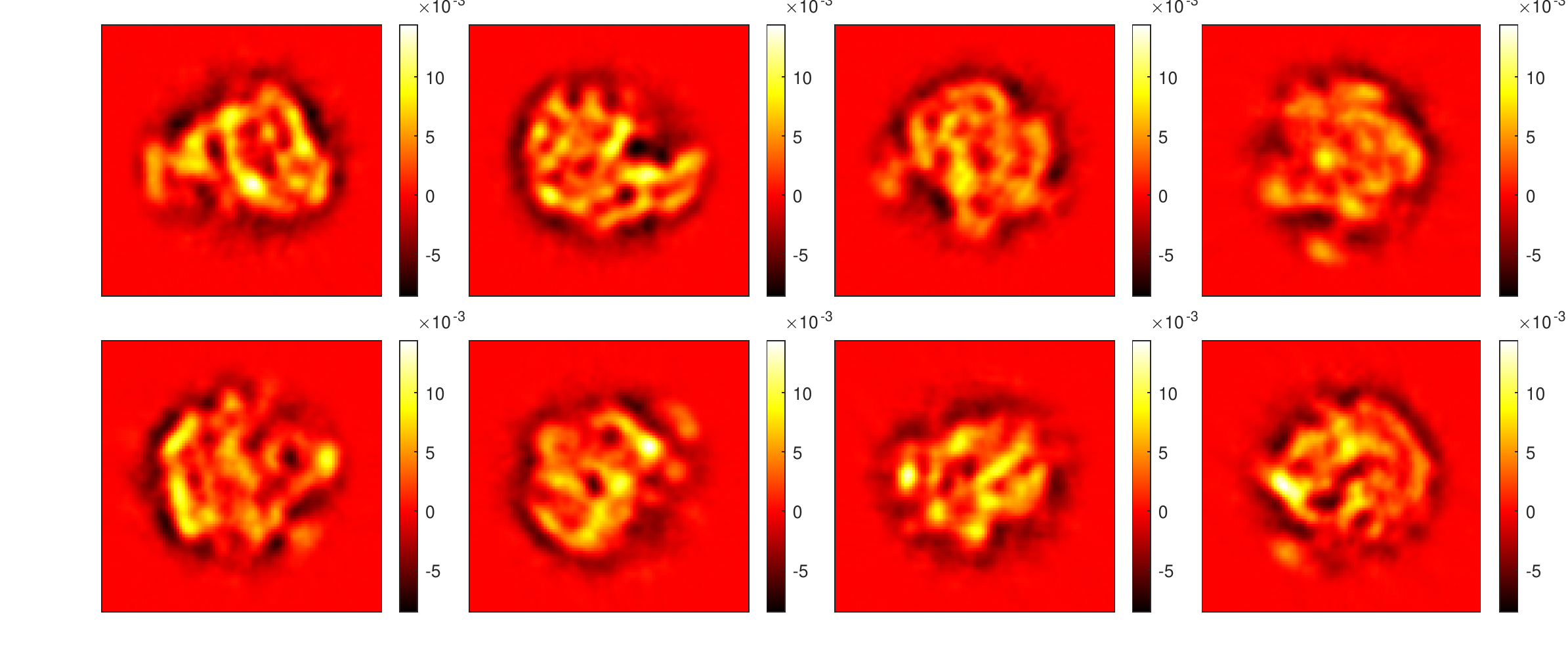}
					
					\caption{Simulated cryo-EM images}
					\label{RandomImageExampleCryo}
				\end{center}
				
			\end{subfigure}
		\end{center}
		\caption{\small
			(a) Examples of random images produced using the procedure described in \Cref{BispectrumComputationSection}.
			(b) Examples of projections in random directions produced using ASPIRE \cite{aspire}.}
		
	\end{figure}
	
	For an image generated as described above, we performed three separate experiments. First, we tested whether its bispectrum is invariant to rotations. The original image was rotated, the rotated image was projected onto the sphere and its coefficients up to bandlimit $ 16 $ were estimated. Its bispectrum was calculated and compared with the bispectrum of the projection of the original image onto the sphere. As \autoref{BispectrumInvarianceFigure} (left) shows, the relative bispecturm error is negligible, and consistent with the numerical perturbations expected the various approximations we use throughout, e.g., our use of interpolation (see \Cref{ProjectionSection}). This is in keeping with \Cref{GroupActionApproxTheorem}\ref{GroupActionApproxTheoremPureRot}.
	
	We next tested whether the bispectrum is invariant to translations. In polar coordinates, every translations is represented by its magnitude and direction. 
	The image was translated by translations of increasing size. For each translation size, the image was translated $ 10^{3} $ times, its bispectrum computed and compared with the bispectrum of the original image. The translation direction was independently, uniformly sampled. \autoref{BispectrumInvarianceFigure} (middle) shows the relative bispectrum error averaged over the translation directions, as well as the symmetric region around the mean where $ 95 $ percent of the samples were. 
	In \autoref{BispectrumInvarianceFigure} (right), we followed a similar procedure, except the image was also rotated. For every translation size, the image was rotated and translated $ 2\cdot 10^{3} $ times. The translation direction and the rotation angle were independently, uniformly sampled. 
	The results indicate the bispectrum is indeed approximately invariant to translations. The error clearly scales with the translation size, as expected from \Cref{GroupActionApproxTheorem}\ref{GroupActionApproxTheoremNonPureRot}.
	Of particular interest to our future numerical experiments, the relative error did not exceed $ 5 $ percent for translation sizes of $ 10 $ pixels (about $ 10 $ percent of the image size) and remained in single digits for even larger translation sizes.
	
	We repeated this experiment with a  simulated cryo-EM image. 
	The image was of size $ 101 \times 101 $ pixels and was simulated as a tomographic projection of a volume in a uniformly random direction using the built in simulation module of the MATLAB package ASPIRE \cite{aspire}.
	Examples of such images appear in \autoref{RandomImageExampleCryo}. 
	Changing the images required increasing the bandlimit in order to accurately represent the projected images on the sphere (see \Cref{ParameterChoiceSection} for details). 
	We used a bandlimit of $ 70 $.
	For a such a large bandlimit, calculating the bispectrum is substantially more time consuming than for bandlimit $ 16 $ we used in the experiment of \autoref{BispectrumInvarianceFigure}.
	Therefore, in order to keep execution time reasonable, in this experiment we sampled much less translation directions, rotation angles in the translation only and translation and rotation experiments (dozens instead of thousands). 
	The results, which appear in \autoref{BispectrumInvarianceCryoFigure}, are substantially the same as the results obtained for a random image, in that the bispectrum is approximately invariant to both rotations and translations, provided the translation size is small enough.
	
	
	\begin{figure}[t]
		\begin{center}
			
			\begin{subfigure}{\textwidth}
				\begin{center}
					\includegraphics[scale=0.55]{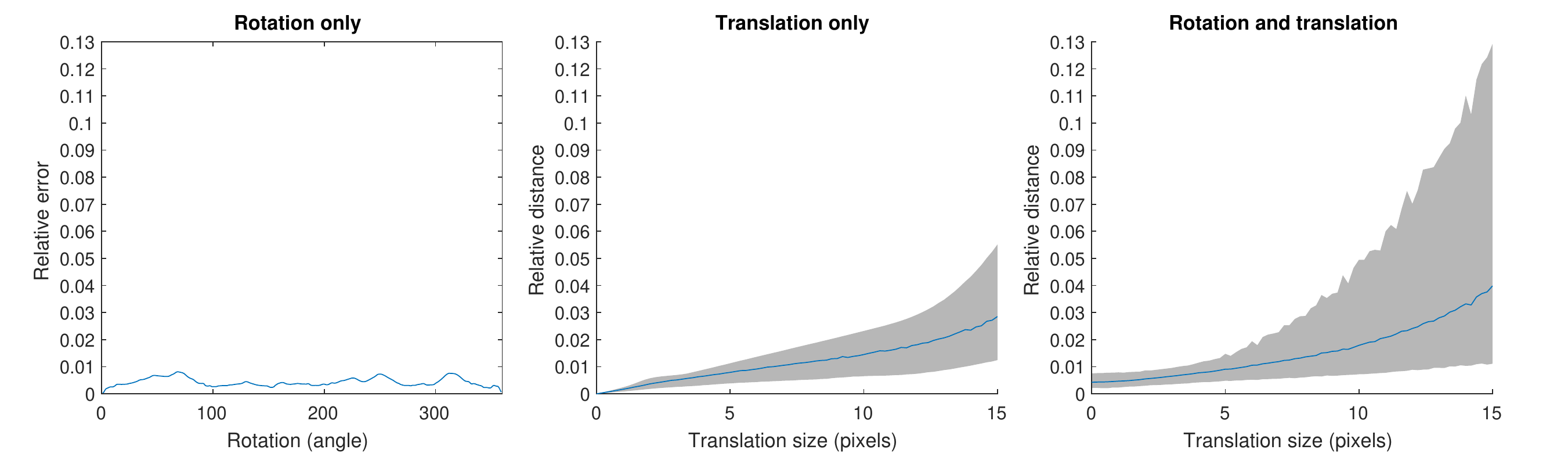}
					
					\caption{Random image from \Cref{BispectrumComputationSection}}
					\label{BispectrumInvarianceFigure}
				\end{center}
			\end{subfigure}
			
			\begin{subfigure}{\textwidth}
				\begin{center}
					\includegraphics[scale=0.55]{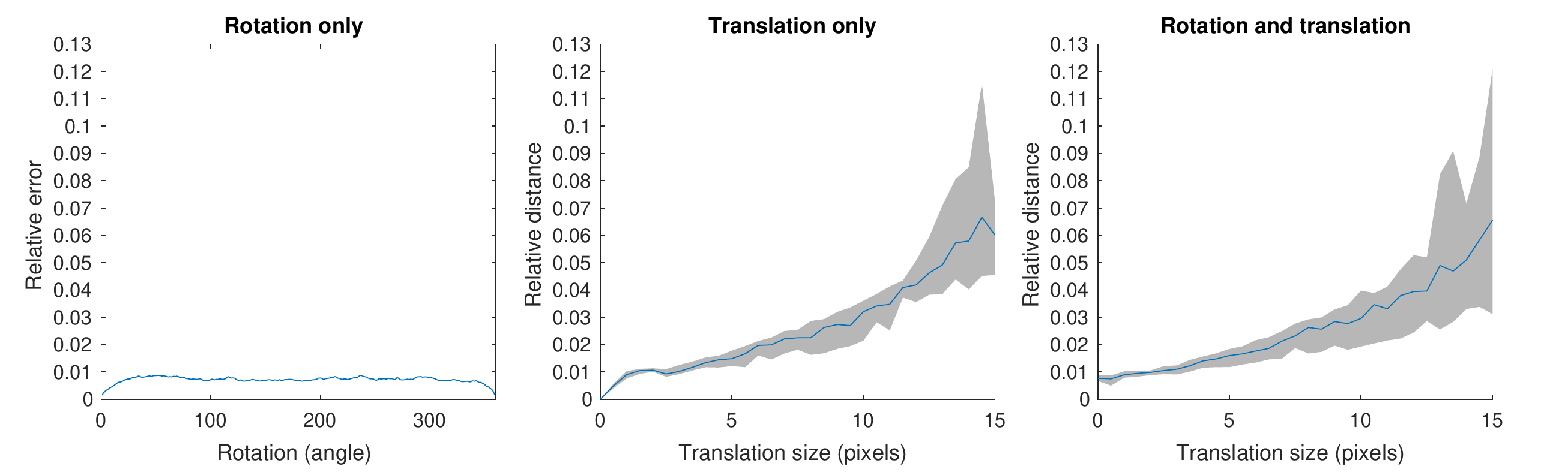}
					
					\caption{Simulated cryo-EM image}
					\label{BispectrumInvarianceCryoFigure}
				\end{center}
				
			\end{subfigure}

		\end{center}
		\caption{\small
			\textbf{\textit{The bispectrum is approximately invariant to rotations and translations of an image.}} The relative bispectrum error is below $ 1 $ percent when the image is only rotated (\textbf{left}) and scales with the translation size when the image is translated only (\textbf{middle}) or translated and rotated (\textbf{right}). The blue curve in the middle and right subfigures is the relative bispecturm error averaged over translation directions and translation directions and rotation angles, respectively. The gray region is a symmetric empirical 95 percent confidence interval around the mean.}
		\label{BispectrumInvarianceFigureLabel}
	\end{figure}
	
	\subsection{Choice of scaling parameters and bandlimit}
	\label{ParameterChoiceSection}
	According to \Cref{GroupActionApproxTheorem} two key parameters, $ \lambda $ and $ \widetilde{\lambda} $, determine how well the orbits of $ \SE(2) $ are approximated by the orbits of $ \SO(3) $.
	The first, $ \lambda $, affects the embedding of the image in $ \F^{2} $, in that it is assumed throughout that the image is compactly supported within $ B_{\lambda \pi} $.
	The second, $ \widetilde{\lambda} $, is the parameter of the contraction map.
	As \autoref{MapVisualDemonstration} illustrates, for a fixed $ \lambda $, increasing $ \widetilde{\lambda} $ maps $ B_{\lambda \pi} $ to a progressively smaller spherical cap around the north pole of the sphere. 
	This implies that a larger $ \widetilde{\lambda} $ makes the fine details of the projected image finer still.
	In order to compensate for this, it is necessary to use a larger bandlimit when computing the projection, because the bandlimit determines the resolution of a spherical harmonics expansion.
	
	We demonstrate this phenomenon empirically.
	We generated a dataset of $ 50 $ images and projected them onto the sphere as described in \Cref{ProjectionSection}, with various choices of scaling parameter and bandlimit. 
	Recall that it is assumed there that every image is a discrete sample on a regular grid of a function, and note that this regular grid is within $ B_{1} $. In particular, it assumed the function is supported within $ B_{\pi} $. 
	In terms of \Cref{GroupActionApproxTheorem}, we fixed $ \lambda = 1 $, as we have in all the numerical experiments throughout the paper, and varied $ \widetilde{\lambda} $.
	For every scaling parameter, the projected images were then projected back onto the grid and the error relative to the ground truth was calculated. 
	For every pair of bandlimit and scaling parameter, the results were averaged over all $ 50 $ images in the dataset.
	\autoref{LambdaChoiceFigure} shows the results for a dataset of random images of the kind we described in \Cref{BispectrumComputationSection} and for a dataset of simulated cryo-EM images.
	As expected, as the scaling parameter increases, a larger bandlimit is required to maintain a similar relative error.
	
	\begin{figure}
		\begin{center}
			\includegraphics[scale=0.58]{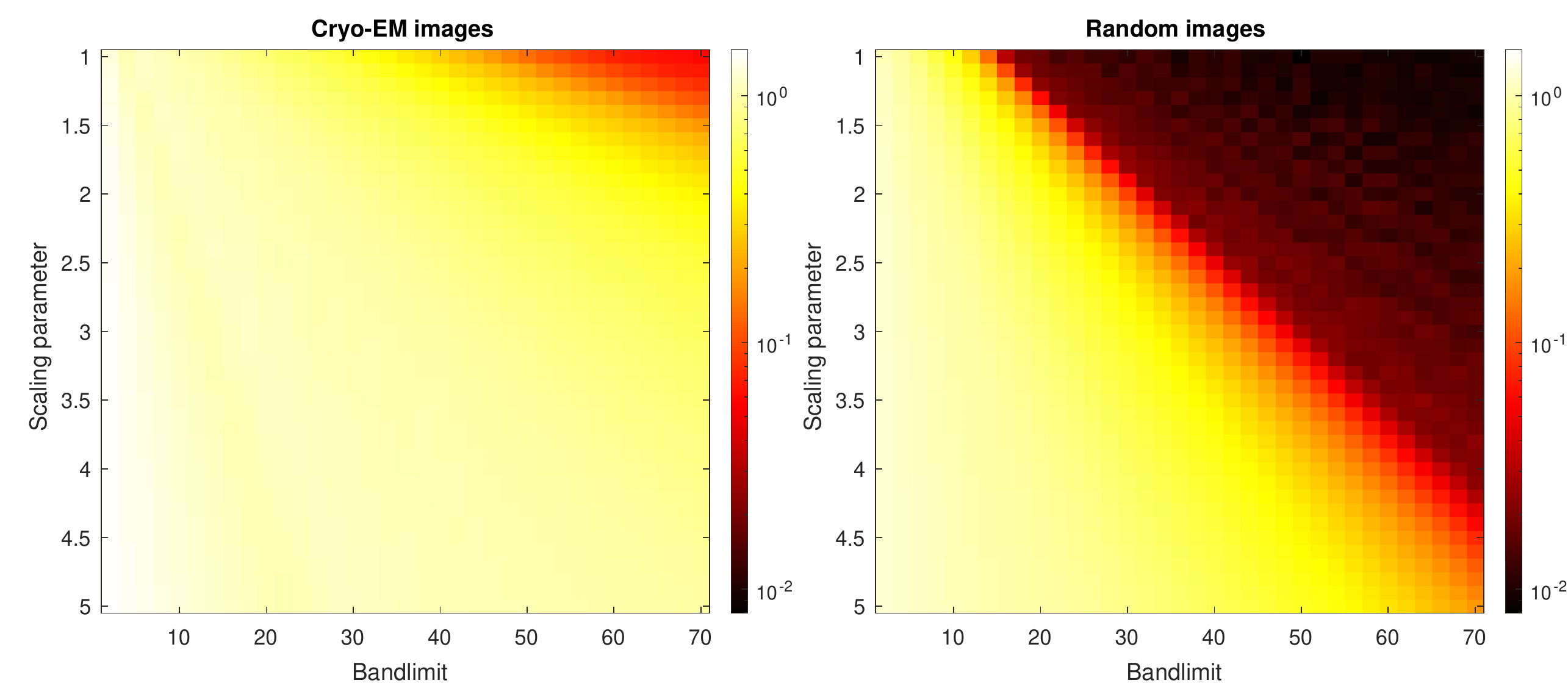}
		\end{center}
		\caption{\small
			\textbf{\textit{Trade-off between the scaling parameter and bandlimit.}} Error of back-projected image relative to the original image averaged over two sets of $ 50 $ images, one with the random images described in \Cref{BispectrumComputationSection} and one from a simulation of cryo-EM.
			As the scaling parameter increases, a larger bandlimit is required to maintain the same relative error. }
		\label{LambdaChoiceFigure}
	\end{figure}
	
	One has to balance between the need to accurately represent projected images on the sphere and the time-complexity of the calculation of the bispectrum. 
	As shown in \autoref{LambdaChoiceFigure}, for bandlimit $ 50 $ and $ 70 $ with scaling parameter of $ 1 $, the mean relative error of recovery for simulated cryo-EM images is about 13 and 6 percent, respectively. 
	However, despite the larger than 10 percent error, for most of our classification experiments on simulated cryo-EM images we used a bandlimit of $ 50 $ and achieved decent results (see \Cref{ClassificationExperimentSection}). 
	This difference is important computationally, since the bispectrum vector has about $ 68  $  thousand elements for bandlimit $ 50 $ and about 182 thousand for bandlimit $ 70 $.
	
	In this paper our concern is demonstrating in principle the viability of compactification in image processing, rather than optimizing it for real-world use. 
	In keeping with this motivation, we chose the bandlimit and the scaling parameter based on empirical criteria.
	For every type of simulated dataset, we sought a bandlimit and scaling parameter pair that fulfilled two requirements. 
	First, it had a decent relative error recovery rate in experiments like the ones showed in \autoref{LambdaChoiceFigure}. For MRA, we generally preferred those with less than 10 percent relative error, but as noted above this could be relaxed for classification.
	Second, it resulted in an approximately rotationally and translationally invariant bispectrum in experiments such as the ones shown in \autoref{BispectrumInvarianceFigureLabel}.
	
	To a large extent, the need to compromise for the sake of computational efficiency is the result of our use of spherical harmonics. 
	On the one hand, using spherical harmonics is convenient. Because of their role in harmonic analysis on the sphere and representation theory of $ \SO(3)  $, they yield a expression for the bispectrum that can be calculated in practice, as discussed in \Cref{BispectrumComputationSection}. This is a property other bases of functions on the sphere generally do not share.
	On the other hand, while our function is localized around the north pole of the sphere, spherical harmonics are generally not. Thus, in order to represent the fine details of an image, we have to take a relatively large bandlimit, resulting in a greater amount of spherical harmonics coefficients needed to represent an image and a greater inefficiency in calculating an invariant representation of them.
	We are currently seeking to address this issue by utilizing a locally supported basis on the sphere (for example, see \cite{Wang2017}).
	
	Finally, we note that in its current form our approach is probably inadvisable for image processing problems involving images with discontinuities or images with non-zero pixels at their boundary.
	Using spherical harmonics to represent the projection of such images onto the sphere is expected to make the trade-off mentioned above worse, as a larger bandlimit will be required to account for the discontinuities.
	A locally supported basis on the sphere might be useful to address this, if it can be used to represent a local feature of an image with higher resolution. Thus, it may allow a more accurate local representation of the discontinuities, without substantially increasing the number of basis elements required to represent the image as a whole.
	In \cite{Wang2017}, a basis capable of such local representations is elaborated upon, and perhaps their approach can be adapted for this puprose. However, we do not intend to pursue this direction here or in future. Our motivation is cryo-EM, where all images can be assumed to be smooth and our main interest here is in demonstration of the potential of our approach, rather than its optimization for real-world use.

	\section{Application I: Compactification in multi-reference alignment over $ \SE(2) $}
	\label{MRASection}
	
	We apply compactifiaction to a multi-reference alignment problem over $ \SE(2) $. Let $ \GrTr : \F^{2} \to \F $ be a smooth function compactly supported within $ \left[-\zeta, \zeta \right]^{2} $, where $ \zeta = \cos \frac{\pi}{4} = \sin \frac{\pi}{4} $. We refer to $ \GrTr $ as the ground truth function. We are given $ N $ discrete samples of $ \GrTr $ of the form:
	\begin{equation}
		\label{ImageFormationModelRecovery}
		\mathbf{I}_{j} 
		= \mathbf{D} \left( g_{j} \bullet \GrTr \right) + \bm{\varepsilon}_{j}, 
		\enskip
		j = 1,2,\dots,N.
	\end{equation}
	Here, $ g_{j} = \left(\mathbf{b}_{j}, R_{j} \right)\in \SE(2) $ is such that $ g_{j} \bullet \GrTr $ is also compactly supported within $ \left[-\zeta, \zeta \right]^{2} $. 
	Furthermore, $ \mathbf{b}_{j} = r_{j} \left(\cos \alpha_{j}, \sin \alpha_{j} \right)^{\top} $ where $ \left\{r_{j}\right\} $ are independently sampled from a uniform distribution on $ \left[0, \MaxTrans \right] $ and $ \alpha_{j} $ is likewise sampled from $ \left[0, 2\pi\right) $.
	Also, $ R_{j} $ is a rotation of the plane by $ \theta_{j} \in \left[0, 2\pi \right) $ radians counterclockwise and $ \left\{\theta_{j}\right\}  $ are independently sampled from a uniform distribution on $ \left[0, 2 \pi \right) $. 
	$ \mathbf{D} $ is a discretization operator sampling a function on a regular equidistant $ n \times n $ grid centered at the origin of $ \F^{2} $. 
	$ \mathbf{I}_{j}  $ and $ \bm{\varepsilon}_{j} $ are $ n \times n $ real matrices representing the samples of the rotated and translated $ \GrTr $ on said grid and noise, respectively. The elements of the noise matrices are Gaussian i.i.d., $ \bm{\varepsilon}_{j,m,k} \sim \mathcal{N} \left( 0, \sigma^{2} \right) $. 
	Stated roughly, our objective is to estimate $ \GrTr $ up to an $ \SE(2) $ action, from sampled images \eqref{ImageFormationModelRecovery}. 
	
	Our approach to the problem is a variant of the invariants approach to MRA described in the introduction. We compactify the samples. 
	By that we mean we project the images onto the sphere to produce bandlimited functions on the sphere, as detailed in \Cref{ProjectionSection}. 
	The spherical bispectrum of each projection is calculated as described in \Cref{BispectrumComputationSection}. 
	We use these to estimate the spherical bispectrum of the projection of $ \GrTr $ onto the sphere and then use it to estimate both the projection of $ \GrTr $ onto the sphere and $\mathbf{D} \GrTr $. 
	
	In the following sections we flesh out this approach.
	In \Cref{DenoisingSection}, we explain how we estimate the spherical bispectrum of the projection of $ \GrTr $ onto the sphere from samples \eqref{ImageFormationModelRecovery}.
	Then, in \Cref{InversionSection} and \Cref{AlignmentSection}, we explain how this estimate is used to estimate both the projection of $ \GrTr $ onto the sphere and $ \mathbf{D} \GrTr $.

	\subsection{Spherical bispectrum estimation}
	\label{DenoisingSection}
	We begin by estimating the spherical harmonics coefficients of the projection of the sampled images onto the sphere and calculating their corresponding bispectrum vectors, using the methods detailed in \Cref{CompactifiactionComputationSection}. We denote the estimate of the former by $ \mathbf{s}_{j} $ and the corresponding bispectrum vector by $ \mathbf{b}_{\mathbf{s}_{j}} $. 
	An unbiased estimator of the spherical bispectrum of the projection of $ \GrTr $ onto the sphere is
	\begin{equation}
		\label{UnbiasedEstimator}
		\widehat{\mathbf{b}} \left[\ell_{1}, \ell_{2}, \ell \right]
		= \frac{1}{N} \sum_{j=1}^{N} \left( \mathbf{b}_{\mathbf{s}_{j}} \left[\ell_{1}, \ell_{2}, \ell\right] - \sum_{k=1}^{3} K_{\mathbf{s}_{j}, k} \left[\ell_{1}, \ell_{2}, \ell  \right]\right),
	\end{equation}
	where 
	$ K_{1} $, $ K_{2} $ and $ K_{3} $ are sparse linear transformations acting on $ \mathbf{s}_{j} $, considered as a real vector. The derivation of \eqref{UnbiasedEstimator} and the full definition of $ K_{\mathbf{s}_{j}, 1} $, $ K_{\mathbf{s}_{j}, 2} $ and $ K_{\mathbf{s}_{j}, 3} $ can be found in \Cref{DebiasingAppendix}.
	Calculating this estimator requires only a single pass on the data. It can also be continuously updated when new data becomes available and therefore is memory-efficient. Its calculation can also be easily parallelized. 
	
	\subsection{Inversion of the spherical bispectrum}
	\label{InversionSection}
	Let $ \mathbf{b} $ be a bispectrum vector of a function on the sphere of bandlimit $ L $. 
	Our objective is to find a spherical function of bandlmit $ L $ which has bispectrum $\mathbf{b}$.
	We do that by solving the unconstrained non-linear least-squares problem
	\begin{equation}
		\label{BispectrumInversionLeastSquares}
		\min_{g \in L^{2}_{L} \left(S^{2} \right)}  \Ltwonorm{ \mathbf{b}_{g} - \mathbf{b}}^{2},
	\end{equation}
	where $ L^{2}_{L} \left(S^{2}\right) $  is the space of spherical functions of bandlimit $ L $ represented by their spherical harmonics coefficients, and $ \mathbf{b}_{g} $ for $ g \in L^{2}_{L} \left( S^{2}\right) $ is its bispectrum vector. In general, \eqref{BispectrumInversionLeastSquares} is a non-convex optimization problem. This is a recurring feature of bispectrum inversion by optimization in various contexts, e.g. \cite{Bendory2018,Ma2020}. Despite that, in the past this approach has worked well and yielded very good estimates. It was also found to be stable in the sense that inverting perturbed bispectrum vectors did not significantly magnify the errors.
	
	We solve \eqref{BispectrumInversionLeastSquares} for our estimator \eqref{UnbiasedEstimator} of the bispectrum of $ F $; that is, $ \mathbf{b} = \widehat{\mathbf{b}}$. We denote the solution by $ \widehat{\mathbf{s}} $ and refer to it as the estimator of $ \kappa_{\lambda} f $ on the sphere, up to 3D rotation. Though we defined the bispecturm vector as a vector of complex numbers, when solving this optimization problem we treat it as a vector of real numbers using the embedding of $ \F[C]^{n} $ into $ \F^{2n} $. Similarly, we solve for the real numbers $ \left\{ \re g_{\ell,m} \right\} \cup \left\{ \im g_{\ell,m}\right\} $, the real and imaginary parts of the spherical harmonics coefficients of $ g $. 
	
	In our implementation, we solve \eqref{BispectrumInversionLeastSquares} using MATLAB's implementation of a trust-regions algorithm in its built in lsqnonlin function. This is an iterative algorithm which requires initialization. 
	In \cite{Ma2020}, a similar approach to ours was used to estimate an image up to 2D rotations only. We project this estimate of the underlying image onto the sphere and use the resulting spherical harmonics coefficients to initialize our the trust-region iteration. 
	In addition, iterative optimization algorithms often perform better when an explicit derivative of the objective function is provided. We therefore provide an explicit expression for the derivative of $ g \mapsto \mathbf{b}_{g} - \widehat{\mathbf{b}} $:
	\begin{equation*}
		\frac{\partial b_{g} [\ell_{1}, \ell_{2}, \ell]}{\partial \re g_{\omega, s}}
		= E_{1} + E_{2} + E_{3}
		\enskip\mbox{and}\enskip
		\frac{\partial b_{g} [\ell_{1}, \ell_{2}, \ell]}{\partial \im g_{\omega,s}}
		= i \left( E_{1} + E_{2} + E_{3} \right),
	\end{equation*}
	where
	\begin{align*}
		E_{1} 
		&= \delta_{\omega, \ell} \delta_{s, m} 
		\sum_{m_{1} = \max\left\{ -\ell_{1}, s - \ell_{2} \right\}}^{\min\left\{ \ell_{1}, s + \ell_{2} \right\}} 
		C_{\ell_{1}, m_{1}, \ell_{2}, s - m_{1}}^{\ell, s} 
		g_{\ell_{1}, m_{1} }^{*} g_{\ell_{2}, s - m_{1} }^{*} \\
		E_{2} 
		&= \delta_{\omega, \ell_{1} } \delta_{s, m_{1}} 
		\sum_{m = \max\left\{ -\ell, s - \ell_{2} \right\}}^{\min\left\{ \ell, s + \ell_{2} \right\}} 
		C_	{\ell_{1}, s, \ell_{2}, m-s}^{\ell, m} 
		g_{\ell, m}  g_{\ell_{2}, s - m_{1} }^{*} \\
		E_{3} 
		&= \delta_{\omega, \ell_{2} } \delta_{s, m_{2}} 
		\sum_{m = \max\left\{ -\ell, s - \ell_{1} \right\}}^{\min\left\{ \ell, s + \ell_{1} \right\}} 
		C_	{\ell_{1}, m-s, \ell_{2}, s}^{\ell, m} 
		g_{\ell, m}  g_{\ell_{1}, m-s }^{*},
	\end{align*}
	where $ \delta_{k, m} $ is the Kronecker delta function defined as $ \delta_{k, m} = 1 $ if $ k = m $ and otherwise $ 0 $.
	These expressions show the derivative is sparse, which reduces both memory and computational requirements.
	
	\subsection{Alignment of two bandlimited functions on the sphere}
	\label{AlignmentSection}
	
	Recall the spherical bispectrum defines a function on the sphere up to 3D rotation. 
	We wish to use our explicit projection formula \eqref{ExplicitProjectionFormula} to project $ \widehat{\mathbf{s}} $ back onto the plane, onto the grid on which the original dataset was defined. However, $ \widehat{\mathbf{s}} $ might not be concentrated around the north pole $ \mathbf{n} = \left(0, 0,1\right)^{\top} $, where said grid is projected to by \eqref{ExplicitProjectionFormula}.
	
	To solve this problem, we allow ourselves to exploit the fact that in our simulations we know the ground truth and align  $ \widehat{\mathbf{s}} $ with the known $ \kappa_{\lambda} F $. 
	This is in keeping with our desire to demonstrate in principle the utility and viability of compactification in MRA problems, rather than optimize it for real-world use. 
	We align the two spherical functions by brute force. We use a deterministic sequence of rotations with ideal uniform distribution in $ \SO(3) $, that was developed in \cite{Yershova2009}. This sequence has the property that $ \left\{R_{i}\right\}_{i=1}^{M} \subseteq \SO(3)$ is uniformly "denser" in $ \SO(3) $ as $ M $ increases. 
	
	For each rotation  $R $ in this sequence, we rotate $ \widehat{\mathbf{s}} $ by $ R $. This is done by rotating a spherical design. We evaluate the spherical function represented by $ \widehat{\mathbf{s}} $ on a spherical $ t $-design with sufficiently large $ t $. The spherical design points are then rotated by $ R $. The resulting set of spherical points is also a spherical $ t $-design \cite{Womersley2017}. We use the values of $ \widehat{\mathbf{s}} $ as the values of the rotated function on the rotated design and estimate its spherical harmonics coefficients.
	We then measure the correlation between the resulting spherical function and the known ground truth spherical function, using the measure: 
	\begin{equation*}
		\mathrm{Corr} (f, g)
		= \frac{
			\re \sum_{\ell=0}^{L} \sum_{m=-\ell}^{\ell} f_{\ell,m} g_{\ell, m}
		}{
			\sqrt{\sum_{\ell=0}^{L} \sum_{m=-\ell}^{\ell} \left| f_{\ell,m} \right|^{2} }
			\sqrt{\sum_{\ell=0}^{L} \sum_{m=-\ell}^{\ell} \left| g_{\ell,m} \right|^{2} }
		}.
	\end{equation*}
	
	From among the first $ M $  rotations in the sequence, we choose the rotation $ R $ that maximizes the correlation and use it to rotate $ \mathbf{s} $ so as to align it to $ \kappa_{\lambda} \GrTr $. 
	In our experiments, we found that using the $ M = 72\cdot 2^{8} $ yields good alignment. 
	Once they are aligned, we evaluate its values on the spherical points corresponding by \eqref{ExplicitProjectionFormula} to the grid on the plane. We denote the resulting image by $ \widehat{\mathbf{I}} $. It is our estimate of the image $ \mathbf{D} \GrTr $. Since we align $ \widehat{\mathbf{s}} $ to the spherical harmonics coefficients of the ground truth image, $ \widehat{\mathbf{I}} $ and $ \mathbf{D} \GrTr $ are already aligned. We measure the error relative to $ \mathbf{D} \GrTr $ in the Frobenius norm.
	
	\subsection{Numerical experiments}
	\label{MRAExperimentSection}
	
	We tested our approach on a synthetic dataset of images. We first generated a ground truth $ 101\times 101 $ image using the procedure described in \Cref{BispectrumComputationSection}. 
	We generated a dataset of images by sampling rotations, translations and noise according to \eqref{ImageFormationModelRecovery} with $ \MaxTrans = 5 $ and applying them to this image. The noise variance $ \sigma^{2} $ was chosen to yield SNR of $ 0.5 $ and is assumed to be known.
	We projected these images onto the sphere with scaling parameter $ \lambda = 1 $ and bandlimit $ 16 $, estimated the bispectrum of the image, inverted it and aligned the resulting function with the projection of the original image onto the sphere. We then projected the aligned function back onto the $ 101 \times 101 $ equidistant grid on the plane and compared it with the original image. This was done for varying sample sizes, repeated $ 15 $ times for each with different noise realization.
	The estimation error was measure using the Frobenius norm for all our estimators.
	
	The results are shown in \autoref{MRASampleSizeFigure}. 
	We first draw the reader's attention to the bispectrum estimator. It was closer to the bispectrum of the original image as the sample size increased. Its slope in a logarithmic scale was $ -0.43 $, which comports decently with the expected $ -0.5 $ slope dictated by the law of large numbers. 
	We also draw the reader's attention to the behavior of our estimator of the image. 
	It improves with an increase in the sample size. 
	In contrast, the initial guess in the bispectrum inversion, which is based on an approach that takes only rotations into account, barely improves as the sample size increases. This is consistent with fact it only takes into account rotations. 
	
	Finally, we note the behavior of our image estimator near what we refer to as the back-projection bound.
	Projecting an image onto the sphere is not lossless. An empirical estimate of the lost information can be obtained by projecting the ground truth image $ \mathbf{D} F  $ from \eqref{ImageFormationModelRecovery} onto the sphere and then back-projecting it back onto the plane. The back-projected image error relative to the original image is a lower bound on the achievable relative error in our estimate of the image.  This lower bound, which we call the back-projection bound, came out in our experiment at about $ 0.02 $. This implies that our method is not expected to be able to estimate the image with less than about $ 2 $ percent relative error.
	Indeed, while our estimator of the image improves with an increase in the sample size, the improvement is approximately linear (in a log-log plot) away from the back-projection bound, but stagnates as the error of the estimator approaches the back-projection bound.
	
	We further tested how our approach responds to varying SNRs. We generated a dataset of images as described above with a fixed sample size $ N = 10^{4} $, and noise realization matching different SNRs. The noise variance $ \sigma^{2} $ is always assumed to be known. For every SNR, we performed $ 15 $ trials. The results appear in \autoref{MRASNRFigure}. All estimators, except the initial guess estimator, performed markedly worse in lower SNRs. While the initial guess is less sensitive to decreasing SNR, it also does not show much improvement when it is increased and is outperformed by our approach when the SNR increases. Combined with the results in \autoref{MRASampleSizeFigure}, it appears our approach may yield a better estimate of the image than the initial guess estimator for lower SNRs, when provided with a sufficiently large dataset. However, such a dataset may be prohibitively large. In order to reduce the required sample size, one might need to refine our approach, perhaps by adding denoising steps prior to projecting the image onto the sphere.

	\begin{figure}[t]
		\begin{subfigure}[t]{0.48\linewidth}
			\includegraphics[scale=0.42]{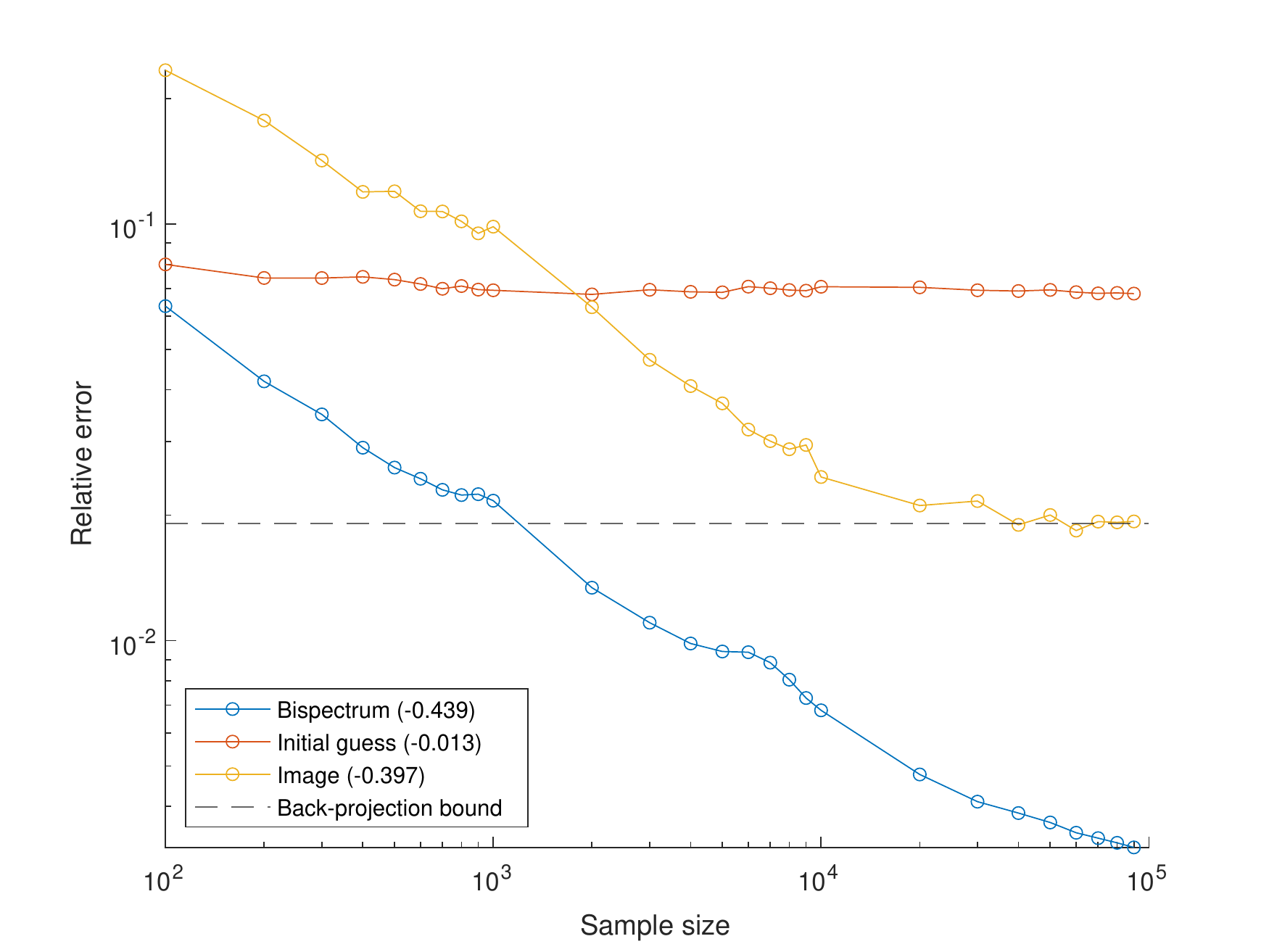}
			\caption{\small
				Relative error vs. sample size
			}
			\label{MRASampleSizeFigure}
		\end{subfigure}
		\begin{subfigure}[t]{0.48\linewidth}
			\includegraphics[scale=0.42]{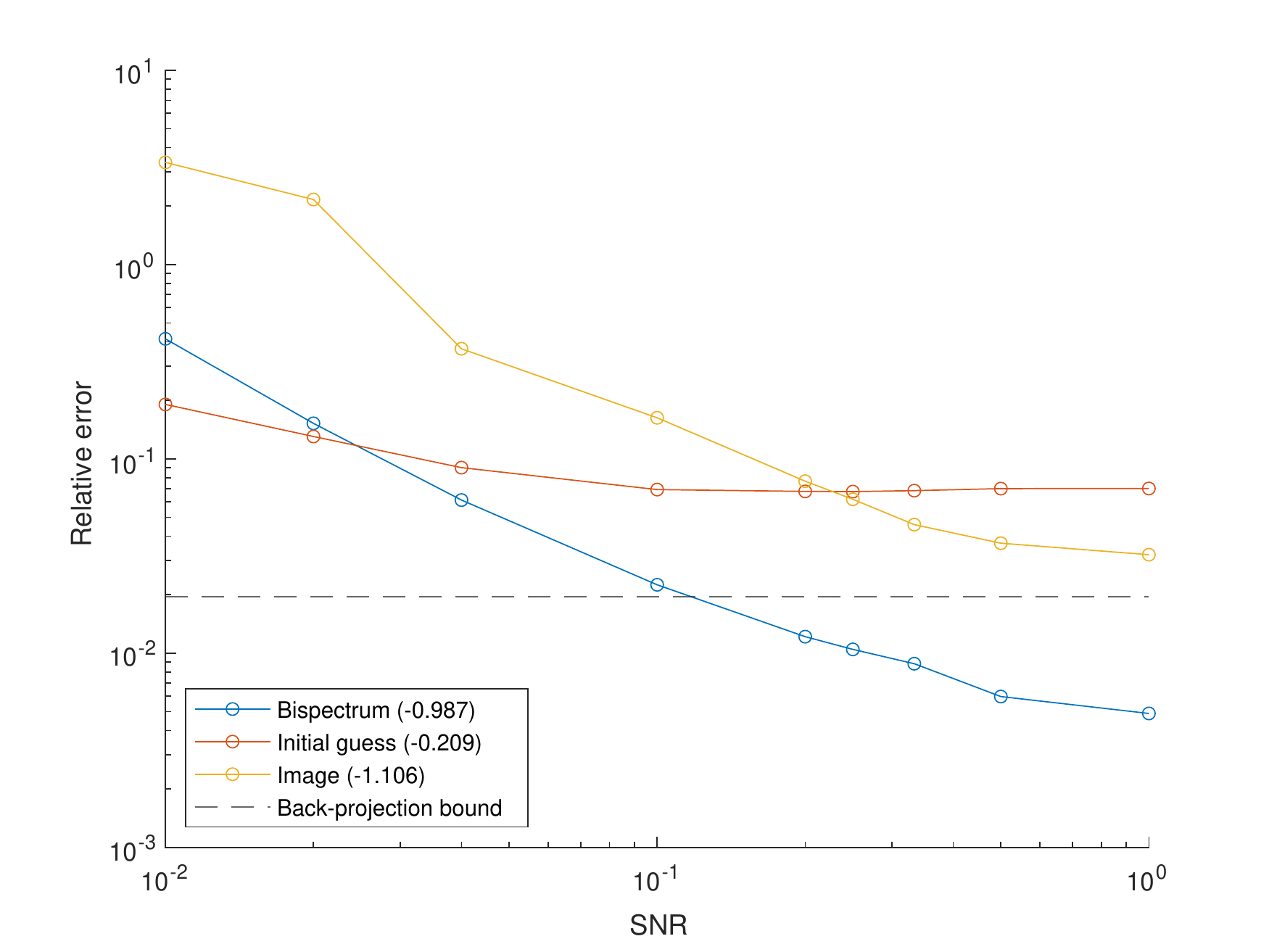}
			\caption{\small
				Relative error vs. SNR}
			\label{MRASNRFigure}
		\end{subfigure}
		
		\caption{\small
			Relative error of estimating the bispectrum and the image using a method of \cite{Ma2020}  which takes only rotations into account (\textbf{Initial guess}), and  using our approach (\textbf{Image}). The \textbf{back-projection bound} is an empirical lower bound on the relative error of our image estimator, which results from a loss of information when projecting an image onto the sphere. All results are averaged over $ 15 $ repeats with different noise realization.
			\textbf{(a)} $ \mathrm{SNR} = 0.5 $ and $ \lambda =1 $. The bispectrum estimator comports decently with the expected behavior of an unbiased estimator. 
			The initial guess was estimated by taking only rotations into account, and shows no scaling with the sample size.
			The image estimator approaches the back-projection bound.
			\textbf{(b)}  Sample size $ 10^{4} $, $ \lambda = 1 $. All estimators perform worse when the SNR is reduced, except the initial guess. It is less sensitive to decreasing SNR, but also does not show much improvement with as it is increased and is outperformed by our approach for large SNRs.}
	\end{figure}

	\section{Application II: Compactification in classification invariant to $ \SE(2) $}
	\label{ClassificationSection}
	
	We apply compactification to the following classification problem invariant to $ \SE(2) $. 
	Each sample belongs to one of $ m $ classes, generated from functions $ \GrTr_{1}, \dots, \GrTr_{m} : \F^{2} \to \F $ compactly supported within $ B_{\pi} $.
	Our samples are assumed to be of the form
	\begin{equation}
		\label{ImageFormationModelClassification}
		\mathbf{I}_{j} 
		= \mathbf{D} \left( g_{j} \bullet \GrTr_{k_{j}} \right) + \bm{\varepsilon}_{j}, 
		\enskip
		j = 1,2,\dots,N, 
		\enskip 
		k_{j} \sim \mathbf{p}.
	\end{equation}
	Here, $ \mathbf{D} $, $ g_{j} $ and $ \bm{\varepsilon}_{j} $ are as described in \eqref{ImageFormationModelRecovery}. 
	Also, $ k_{j} $ are drawn i.i.d.\ from  a distribution $ \mathbf{p} = \left(p_{1},\dots,p_{m} \right)^{\top}$. 
	Finally, given a dataset of the form \eqref{ImageFormationModelClassification}, let $ L(j) = \GrTr_{k_{j}} $ be the label of the $ j $th image. 
	Our objective is to identify for every image its $ K $-closest images up to rotation and translation.
	
	In Zhao and Singer \cite{Zhao2014}, a similar problem was dealt with in the context of cryo-EM. There a method that takes only rotations into account was used. In our approach and theirs, every image in a dataset is transformed into a rotation-invariant representation. We describe a similar representation to theirs in \Cref{RotationInvariantRepresentationSection} using the approximately $ \SE(2) $-invariant representation we developed in \Cref{CompactifiactionTheorySection,CompactifiactionComputationSection}. A $ K $-nearest neighbors graph is then constructed by measuring the distance between the invariant representation of the images, rather than the images themselves, as detailed in \Cref{GraphConstructionSection}. 
	
	In \Cref{ClassificationExperimentSection}, we numerically test how our approach, which takes into account both rotations and translations, performs compared to an approach that takes only rotations into account. We show that the similarity graphs match the classification function of the dataset better when translations are taken into account, even for relatively large maximal translation size $ \MaxTrans $ in \eqref{ImageFormationModelClassification}.
	
	\subsection{Steerable basis and a rotation-invariant representation of a function}
	\label{RotationInvariantRepresentationSection}
	
	We use the MATLAB package ASPIRE \cite{aspire} to estimate the expansion coefficients of $ g_{j} \bullet F_{k_{j}} $ in a truncated Fourier-Bessel basis \cite{Zhao2016, Zhao2013}. Every member of this basis has the form
	\begin{equation*}
		u_{k,q} (r,\phi)
		= f_{k,q} (r) e^{i k \phi},
	\end{equation*}
	where $ (r,\phi) $ are elements of $ \F^{2} $ in polar coordinates, $ f_{k,q} $ is called the radial part and $ \phi \mapsto e^{i k \phi} $ is called the angular part. 
	We assume all functions we expand in this basis are bandlimited in the sense they have a finite expansion. In particular, $ \left|k\right| \le k_{\mathrm{max}} $ and $ 0 \le q \le Q_{k}  $.
	
	A basis with this form is said to be steerable, meaning that it diagonalizes the action of $ \SO(2) $. In particular, if $ R = R(\varphi) $ is the rotation of $ \F^{2} $ by $ \varphi $ counter-clockwise and
	\begin{equation}
		\label{FourierBesselExpansion}
		F 
		= \sum_{k,q} F_{k,q} u_{k,q},
	\end{equation}
	then
	\begin{equation*}
		R \bullet F (r,\theta)
		= F(r, \theta + \varphi)
		= \sum_{k,q} F_{k,q} f_{k,q} (r) e^{i k (\theta+\varphi)}
		= \sum_{k,q} e^{ik\varphi} F_{k,q} f_{k,q} (r) e^{i k \theta}.
	\end{equation*}
	Thus, the expansion of $ R \bullet F $ in the Fourier-Bessel basis satisfies $ (R\bullet F)_{k,q} = e^{i k \varphi} F_{k,q} $.
	
	In \cite{Zhao2014}, this property was used to define a rotation-invariant representation of a function $ F $ of the form \eqref{FourierBesselExpansion}:
	\begin{equation}
		\label{RotationBispectrum}
		b_{F} [k_{1}, k_{2}, q_{1}, q_{2}, q_{3}]
		= F_{k_{1},q_{1}}F_{k_{2},q_{2}} F_{k_{1}+k_{2},q_{3}}^{*},
	\end{equation}
	for all tuples of radial and angular indices $ (k_{1}, k_{2}, q_{1}, q_{2}, q_{3}) $ satisfying
	\begin{equation}
		\label{RotationBispectrumIndices}
		\left|k_{1} + k_{2}\right| \le k_{\mathrm{max}}, \ 
		0 \le q_{1} \le Q_{k_{1}}, \ 
		0 \le q_{2} \le Q_{k_{2}}
		\mbox{ and }
		0 \le q_{3} \le Q_{k_{1} + k_{2}}.
	\end{equation}
	In keeping with the terminology of \cite{Zhao2014}, we refer to \eqref{RotationBispectrum} as the rotational bispectrum. 
	When $ F $ is a Fourier-Bessel expansion estimated from an image, we refer to \eqref{RotationBispectrum} as the rotational bispectrum of that image.
	
	We wish to represent every image in our dataset using its rotational bispectrum. Unfortunately, even when $ k_{\mathrm{max}} = 5 $, the number of tuples satisfying \eqref{RotationBispectrumIndices} is extremely large, on the order of tens of thousands at least and grows to millions even for modest $ k_{\mathrm{max}} $.
	Therefore, even for modestly sized datasets of images, we cannot store in memory the rotational bispectrum of all images. On the other hand, the size of the rotational bispectrum suggests it contains a lot of redundancy. Thus, like in \cite{Zhao2014}, we apply dimensionality reduction to the collection of rotational bispectra of all images in a dataset.
	
	In particular, let $ \mathbf{b}_{j} $ denote the rotational bispectrum of the $ j $th image. Let $ \mathbf{B} $ be a matrix with $ N $ columns, the $ j $th column of which is the rotational bispectrum of $ \mathbf{I}_{j} $ in lexicographical order of the indices \eqref{RotationBispectrumIndices}. We apply to $ \mathbf{B} $ a randomized algorithm, designed to approximate the best rank $ m $ approximation of a large matrix without forming it in memory \cite{Halko2011}. 
	We use it to approximate $ \mathbf{B} $ by a $ m \times N $ dimensional matrix.
	We refer to the $ j $th column of the resulting matrix either as the reduced rotational bispectrum of $ \mathbf{I}_{j} $ or as its reduced rotation-invariant representation.
	
	\subsection{Constructing a similarity graph from invariant representations}
	\label{GraphConstructionSection}
	
	Given a dataset of images of the form \eqref{ImageFormationModelClassification}, assume we represent it either by the reduced rotation-invariant representation introduced in \Cref{RotationInvariantRepresentationSection} or by the $ \SE(2) $-invariant representation introduced in \Cref{CompactifiactionComputationSection}, possibly after going through denoising by a linear filter. Denote by $ \widetilde{\mathbf{I}}_{j} $ the invariant representation of $ \mathbf{I}_{j} $. We define two invariant distance metric on our dataset by 
	\begin{equation}
		\label{InvariantMetric}
		d \left(\mathbf{I}_{j_{1}}, \mathbf{I}_{j_{2}} \right)
		= \Ltwonorm{\widetilde{\mathbf{I}}_{j_{1}} - \widetilde{\mathbf{I}}_{j_{2}} }_{F}.
	\end{equation}
	When $ \widetilde{\mathbf{I}}_{j_{1}} $ and $ \widetilde{\mathbf{I}}_{j_{2}} $ are the reduced rotation-invariant representations of the corresponding images, we refer to this metric as the rotation-invariant metric. Similarly, when these are our approximately $ \SE(2) $-invariant representations, we refer to it as the $ \SE(2) $-invariant metric.
	
	We use an invariant distance metric to find the $ K $-nearest images of every image in our dataset. Denote by $ G_{j} $ the set of indices of the $ K $ images closest to $ \mathbf{I}_{j} $. For every node, we then measure the proportion of its $ K $ neighbors belonging to its class. That is, we calculate:
	\begin{equation}
		\label{NodeScore}
		s_{j}
		= \frac{\left|\left\{ j_{1} \in G_{j}  \setsep j_{1} \ne j,\ L (j_{1}) = L(j)   \right\}\right|}{K}.
	\end{equation}
	We refer to \eqref{NodeScore} as the node score of the $ K $-nearest neighbors graph we construct using the metric \eqref{InvariantMetric}. We consider its distribution as a measure of the quality of the classification. The classification is considered better the more concentrated the distribution of node score is close to $ 1 $.

	\subsection{Numerical experiments}
	\label{ClassificationExperimentSection}
	
	We tested our approach numerically as follows. We first generated seven class representatives using the procedure described in \Cref{BispectrumComputationSection}. We generated a dataset of $ 5\cdot 10^3 $ $ 101\times 101 $ images following \eqref{ImageFormationModelClassification} with $ m = 7 $ and $ p_{1} = p_{2} = \dots = p_{m} = \frac{1}{7} $. 
	We then generated two $ 50 $-nearest neighbors graphs, one using the rotation-invariant metric described in \Cref{RotationInvariantRepresentationSection} and another using our $ \SE(2) $-invariant metric. For each graph we measured the node score \eqref{NodeScore} of every node. 
	We performed this procedure for  maximal translation sizes $ \MaxTrans = 0, 2.5, 5, 7.5, 10 $. 
	The labels, the rotation angle, the translation directions and the translation sizes of the images in our dataset were sampled in advanced, the latter was sampled uniformly from $ [0,1] $ and then rescaled by multiplying it by the appropriate $ \MaxTrans $. 
	For every $ \MaxTrans $, we repeated this experiment $ 10 $ times, each with different noise realization, all with $ \mathrm{SNR} = 1$. The noise variance $ \sigma^{2} $ is always assumed to be known.
	For every $ \MaxTrans $, we averaged the node score over the $ 10 $ noise realizations and plotted the histogram of the average.
	
	\autoref{ClassificationFigure} shows the histogram of the average node score for varying maximum translation sizes $ \MaxTrans $, when using a rotation-invariant metric and our $ \SE(2) $-invariant metric. The more the histogram is concentrated close to $ 1 $, the better the distribution of node score is. Both metrics yield similar results in the absence of translations. However, our metric is far less sensitive to the translation size and even yields excellent node score distribution for $ \MaxTrans  =10 $.
	
	We repeated this experiment with simulated cryo-EM projections of the type shown in \autoref{RandomImageExampleCryo}. 
	Each class representative was generated using the built in simulation module of the MATLAB package ASPIRE \cite{aspire}.
	In accordance with the results of \Cref{ParameterChoiceSection}, a bandlimit of $ 50 $ was chosen.
	Because of the computational load of using such a large bandlimit, the experiment was done with  a single noise realization.
	For comparison, we present the node score obtained by the algorithm of \cite{Zhao2014}, which takes only rotations into account. 
	Specifically, we used an optimized implementation of it which is included in ASPIRE \cite{aspire}, which also includes a denoising stage via steerable PCA \cite{Zhao2016}.
	As the results in \autoref{ClassificationCryoFigure} indicate, classification using our approximate invariant was more accurate for simulated cryo-EM images.
	
	Finally, we repeated this experiment once on a dataset that resembles more the kind of datasets encountered in cryo-EM applications. 
	We used a dataset of $ 10^{4} $ images of size $ 101 \times 101 $. The maximal translation size was set to $ 10 $ pixels, there were $ 100 $ class representatives, rather than seven, and the bandlimit used in our algorithm was $ 70 $.
	We again compared the distribution of node score achieved by our approach with that obtained by the same optimized implementation of the algorithm of \cite{Zhao2014} in ASPIRE \cite{aspire}.
	The results are shown in \autoref{fig:ClassificationRealistic}.
	Our algorithm clearly had a better distribution of node score. Consistent with that, our approach yielded a median node score of $ 1 $, compared with $ 0.22 $ using the rotations-only invariant. 
	
	\begin{figure}
		
		\begin{center}
			\begin{subfigure}{\textwidth}
				\begin{center}
					\includegraphics[scale=0.6]{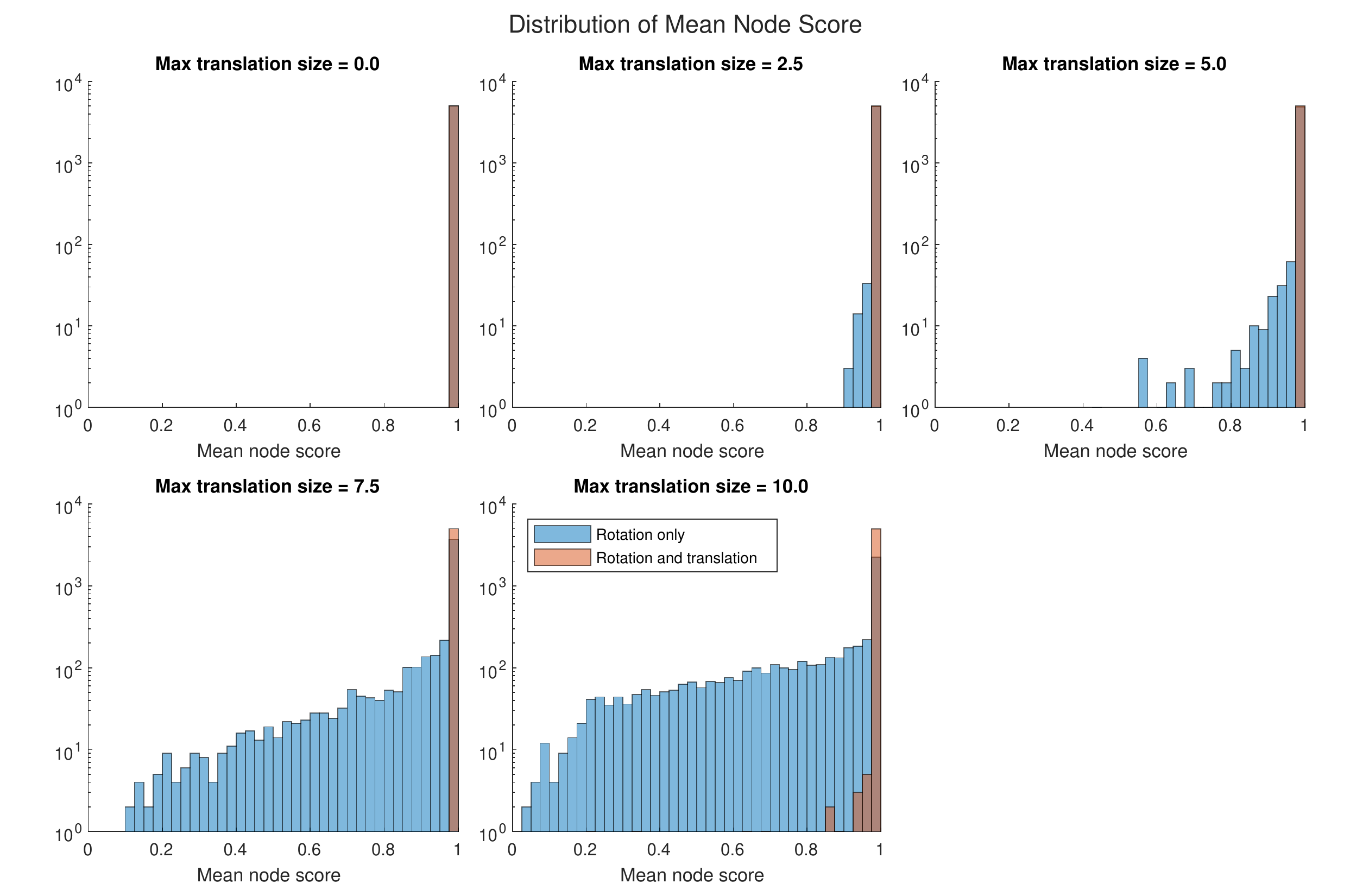}
					
					\caption{Random images from \Cref{BispectrumComputationSection}}
					\label{ClassificationFigure}
				\end{center}
			\end{subfigure}
			
			\begin{subfigure}{\textwidth}
				\begin{center}
					\includegraphics[scale=0.6]{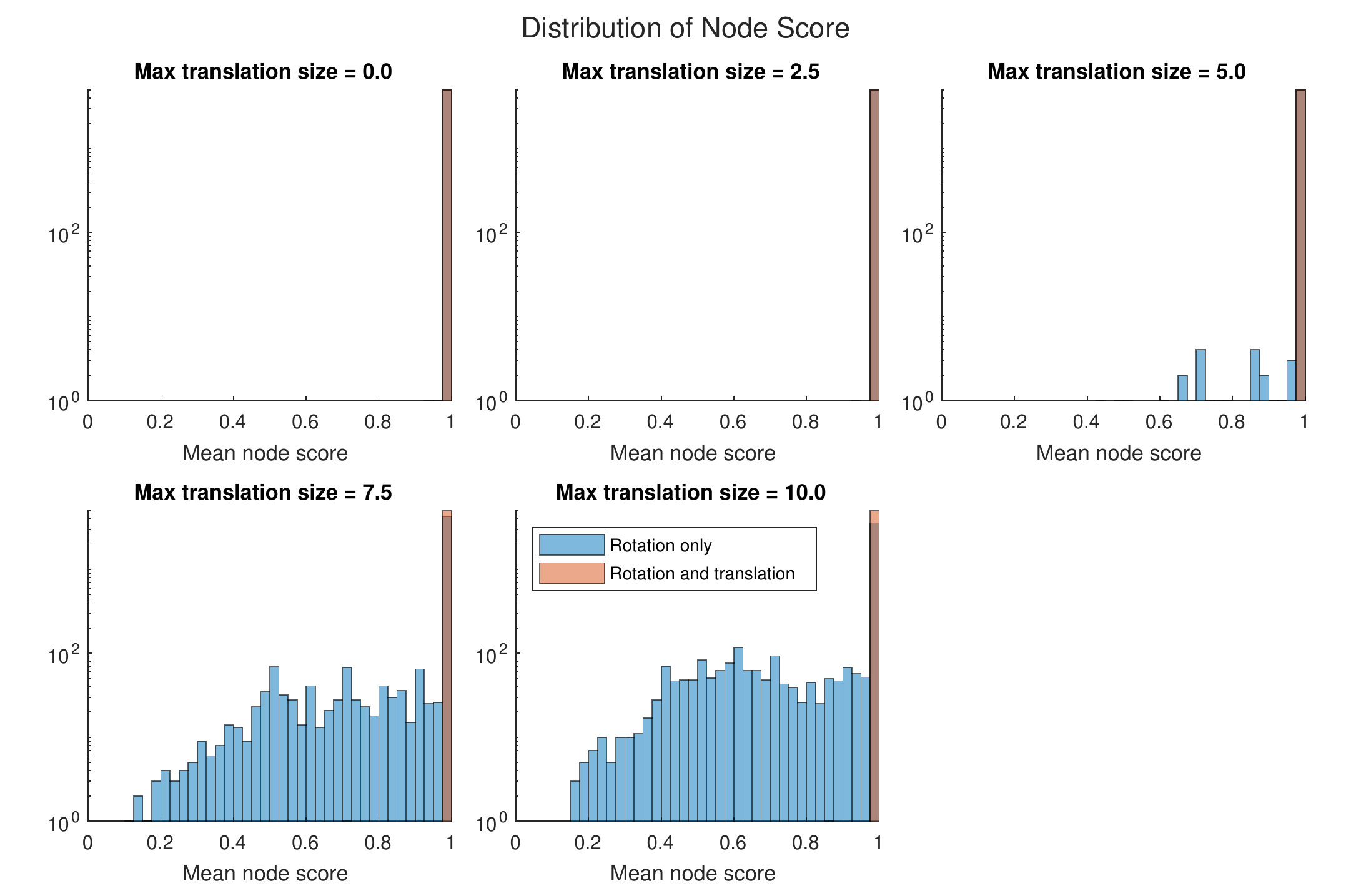}
					
					\caption{Simulated cryo-EM images}
					\label{ClassificationCryoFigure}
				\end{center}
				
			\end{subfigure}
		\end{center}
		
		\caption{\small
			\textbf{\textit{Node score for rotationally invariant and approximately rotationally and translationally invariant metrics.}} 
			A histogram of the node score for maximal translation sizes $ \MaxTrans = 0, 2.5, 5, 7.5, 10 $ on a dataset generated from seven class representatives. For (a), the node score is averaged over $ 10 $ different noise realization. The data in (b) is based on a single noise realization.
			Both metrics yield similar results in the absence of translations. Taking into account translations significantly reduces the sensitivity of the distribution of node score to translations.}
	\end{figure}
	
	\begin{figure}
		
		\begin{center}
			\includegraphics[scale=0.6]{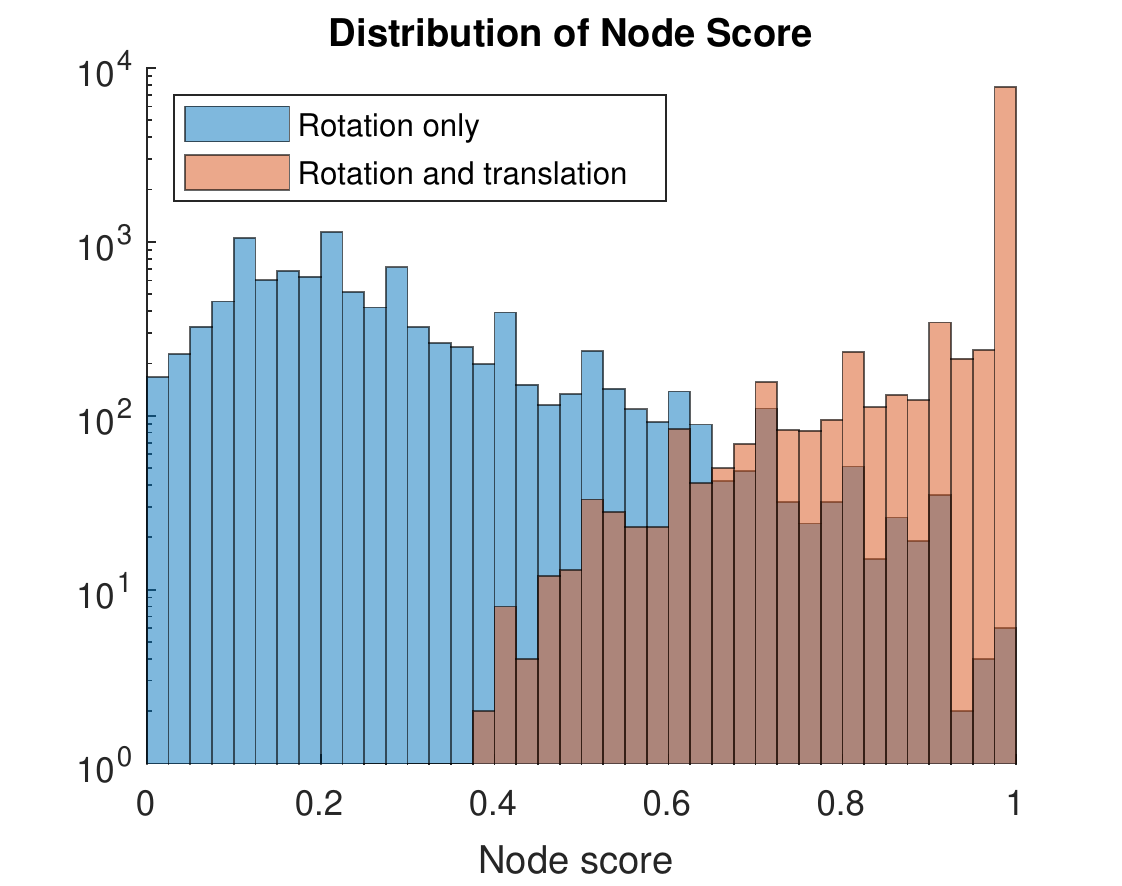}
		\end{center}
		
		\caption{\textbf{\textit{Node score for rotationally invariant and approximately rotationally and translationally invariant metrics on a dataset of simulated cryo-EM images containing a hundred class representatives.}} A histogram of the node score for maximal translation size $ 10 $ on a dataset generated from a hundred class representatives. Taking into account translations yielded a median node score of $ 1 $, whereas using a rotations-only invariant yielded a median node score of $ 0.22 $.}
		\label{fig:ClassificationRealistic}
	\end{figure}
	
	\appendix
	
	\section{Proofs for \Cref{CompactifiactionTheorySection}}
	The various lemmas were not stated in the main text in the order in which they're proved. 
	In particular, the proofs make it clear that 
	\Cref{SmoothSubmersionLemma} and \Cref{SmoothSubmersionLemmaFibers}
	depend on \Cref{ExplicitSubmersionLemma}.
	We altered the order in the main text to improve the exposition.
	
	\label{ProofsAppendix}
	\begin{proof}[Proof of \Cref{ExplicitSubmersionLemma}]
		Recall the identity \eqref{EtaSOPsiLambdaFormula}.
		Since $ \so(3) $ is a Lie subalgebra of $ \mathfrak{gl} (3) $, one can write for every $ \mathbf{x} \in \F^{2} $:
		\begin{equation}
			\label{expb}
			\exp_{\SO(3)} (\mathbf{x})
			= \sum_{n=0}^{\infty} \frac{1}{n!} \mathbf{X}^{n}, 
			\qquad \mbox{where } 
			\mathbf{X}
			= \left[\begin{matrix}
				\mathbf{0}_{2 \times 2} & \mathbf{x} \\
				- \mathbf{x}^{\top} & 0
			\end{matrix}\right].
		\end{equation}
		From this point, the proof relies primarily on the Taylor series of sine and cosine.
		Note that
		\begin{align*}
			\mathbf{X} \mathbf{n} 
			= \left(\begin{matrix}
				\mathbf{x} \\ 0
			\end{matrix}\right) 
			\quad\mbox{and}\quad
			\mathbf{X}^{2} \mathbf{n}
			= \mathbf{X} 
			\left(\begin{matrix}
				\mathbf{x} \\ 0
			\end{matrix}\right)
			= - \Ltwonorm{\mathbf{x}}^{2} \mathbf{n} .
		\end{align*}
		Therefore,
		\begin{equation*}
			\mathbf{X}^{2n} \mathbf{n} 
			= (- 1)^{n} \Ltwonorm{\mathbf{x}}^{2n} \mathbf{n}
			\quad \mbox{and} \quad
			\mathbf{X}^{2n+1} \mathbf{n} 
			= \left( - \Ltwonorm{\mathbf{x}}^{2}\right)^{n} \mathbf{X} \mathbf{n}
			= (-1)^{n} \Ltwonorm{\mathbf{x}}^{2n}
			\left(\begin{matrix}
				\mathbf{x} \\ 0
			\end{matrix}\right).
		\end{equation*}
		This implies the $ z $-coordinate of \eqref{expb} is
		\begin{equation}
			\label{expb1}
			\sum_{n=0}^{\infty} \frac{ (-1)^{n} \Ltwonorm{\mathbf{x}}^{2n} }{(2n)!}
			= \cos \Ltwonorm{\mathbf{x}},
		\end{equation}
		and the $ x $- and $ y $-coordiantes are of the form
		\begin{equation}
			\label{expb2}
			x \sum_{n=0}^{\infty} \frac{(-1)^{n} \Ltwonorm{\mathbf{x}}^{2n}}{(2n+1)!}
			= \frac{x}{\Ltwonorm{\mathbf{x}}} \sum_{n=0}^{\infty} \frac{(-1)^{n} \Ltwonorm{\mathbf{x}}^{2n+1}}{(2n+1)!}
			= \frac{x}{\Ltwonorm{\mathbf{x}}} \sin \Ltwonorm{\mathbf{x}},
		\end{equation}
		where $ x $ is the $ x $-coordinate or the $ y $-coordinate of $ \mathbf{x} $, respectively.
		To conclude the proof, substitute $ \mathbf{x} $ with $ \mathbf{x}/\lambda $ in \eqref{expb1}  and \eqref{expb2}.
	\end{proof}
	
	\begin{proof}[Proof of \Cref{SmoothSubmersionLemma} and \Cref{SmoothSubmersionLemmaFibers}]
		$ \SubmerSE $ is a surjective map and it is an orbit map of $ \mathbf{0}_{2} $ of the group action of $ \SE(2) $ on $ \F^{2} $; that is, $ \SubmerSE(\mathbf{x}, R) = (\mathbf{x}, R) \bullet \mathbf{0}_{2} $.
		It is therefore a smooth submersion (cf. \cite[Prop. 7.26, p. 83]{Lee2013} and \cite[Thm. 4.14, p. 166]{Lee2013}). 
		In order to prove $ \SubmerSO \circ \Psi_{\lambda} $ is a smooth submersion, define $ P (\zeta, \theta, \phi) = \left( \zeta \cos \theta, \zeta \sin \theta, R(\phi) \right)^{\top} $ as a mapping $ P : \F^{3} \to \SE(2) $, where $ R(\phi) $ is a rotation of the plane by $ \phi $ radians counterclockwise. Also define the standard spherical coordinates transformation $ S (\varphi, \delta) = \left(\cos \varphi \sin \delta, \sin \varphi \sin \delta, \cos \delta \right)^{\top} $ as a mapping $ S : \F^{2} \to S^{2} $. $ P $ and $ S $ are surjective local diffeomorphisms.
		
		By \Cref{ExplicitSubmersionLemma}, 
		\begin{equation*}
			\SubmerSO \circ \Psi_{\lambda} \circ P ( \zeta, \theta, \phi ) 
			= \left(\cos \theta \sin \left( \frac{\zeta}{\lambda} \right), 
			\sin \theta \sin \left( \frac{\zeta}{\lambda} \right), 
			\cos \left( \frac{\zeta}{\lambda} \right) \right)^{\top} .
		\end{equation*}
		Since $ S $ is a local diffeomorphism about $ S^{2} $, it is locally invertible around every point of $ S^{2} $. This implies
		\begin{equation}
			\label{LinearExplicitSmoothSubmersion}
			S^{-1 } \circ \SubmerSO \circ \Psi_{\lambda} \circ P (\zeta, \theta, \phi )
			= \left( \theta, \frac{\zeta}{\lambda} \right)^{\top}.
		\end{equation}
		Thus, $ S^{-1} \circ \SubmerSO \circ \Psi_{\lambda} \circ P (\zeta, \theta, \phi) $ is locally defined and locally linear, implying $ \SubmerSO \circ \Psi_{\lambda} $ is a smooth submersion. 
		
		Finally, \eqref{LinearExplicitSmoothSubmersion}, the properties of the spherical coordinates and polar coordinates easily yield the characterization of the fibers of $ \SubmerSO \circ \Psi_{\lambda} $ in \eqref{Fiber1}, \eqref{Fiber2} and \eqref{Fiber3}.
	\end{proof}
	
	\begin{proof}[Proof of \Cref{ExtensionLemma}]
		Since $ f $ is a smooth function on a homogeneous space of $ \SE(2) $, we can lift it to a function on $ \SE(2) $ that is constant on left cosets of $ \SO(2) $. Thus, without loss of generality, assume $ f : \SE(2) \to \F $.  Now note that the fibers of $ \SubmerSO \circ \Psi_{\lambda} $ of forms \eqref{Fiber1} and \eqref{Fiber3} intersect $ B_{ \lambda \pi } \times \SO(2) $ only once and that the fiber of the form \eqref{Fiber2} only intersects $ \partial B_{ \lambda \pi } \times \SO(2) $. Define $ \widetilde{f} $ on fibers of the form \eqref{Fiber1} and \eqref{Fiber3} to be the value of $ f $ on their intersection with $ B_{ \lambda \pi } \times \SO(2) $ and zero on fibers \eqref{Fiber2}. 
		
		The uniqueness of $ \widetilde{f} $ is immediate from our construction. It remains to prove $ \widetilde{f} $ is smooth. 
		Let $ B^{(n)} \coloneqq \left(B_{\lambda \pi n}\setminus \{0\}\right) \setminus \overline{B}_{\lambda \pi (n-1)}  $. 
		Define a map $ p : B_{\lambda \pi }\setminus \{0\} \to B^{(n)} $ as $ p(\zeta, \theta) = \left( \zeta + \lambda \pi  (n-1), \theta\right) $. This is obviously a (local) diffeomorphism. Note that for $ (\mathbf{x}, R) \in B^{(n)} \times \SO(2) \subset \SE(2) $ we have $ \widetilde{f} (\mathbf{x}, R) = f \circ p^{-1} (\mathbf{x}) $. Since both $ f $ and $ p $ are smooth, it follows that $ \widetilde{f} $ is smooth in a neighborhood of $ \mathbf{x} $. It remains to prove $ \widetilde{f} $ is smooth for all $ (\mathbf{x}, R) \in \partial B_{\lambda \pi n} \times \SO(2) $. Since $ f $ is compactly supported in $ B_{\lambda \pi } $, it follows that it is identically zero in a neighborhood $ V $ of $ \mathbf{x} $. Therefore, making $ V $ smaller if necessary, we conclude that $ \widetilde{f} $ is identically zero on $ V \times \SO(2) $ and so it is smooth there as well.
	\end{proof}

	\begin{proof}[Proof of \Cref{LipschitzStuff}]
		Recall the identity $ 2 \sin^{2} \left( \frac{x}{2} \right) = 1 - \cos x $ and the inequality $ \sin x \ge \frac{2}{\pi} x $ for all $ x \in [0,\pi/2] $. Combining them it follows that
		\begin{equation}
			\label{UtilityIneq}
			2 (1 - \cos x)
			= 4 \sin^{2} \left( \frac{x}{2} \right) 
			\ge 4 \cdot \left( \frac{2}{\pi} x \right)^{2}
			= \frac{16}{\pi^2} x^{2}
			\ge \frac{x^{2}}{\pi}.
		\end{equation}
		Now, let $ \mathbf{x}, \mathbf{y} \in S^{2} $ and denote $ \theta = d_{S^{2}} (\mathbf{x}, \mathbf{y}) $. Thus:
		\begin{align*}
			\Ltwonorm{\mathbf{x} - \mathbf{y}}^{2} 
			&= \Ltwonorm{\mathbf{x}}^{2} + \Ltwonorm{\mathbf{y}}^{2} - 2 \mathbf{y}^{\top} \mathbf{x} \\
			&= 2 (1 - \cos \theta) && \mbox{Because $ x, y\in S^{2} $} \\
			&\ge \frac{\theta^{2}}{\pi}  && \mbox{From \eqref{UtilityIneq}.}
		\end{align*}
		The result easily follows.
	\end{proof}
	
	\begin{proof}[Proof of \Cref{ApproximateCumm1}]
		Let $ A $ and $ B $ be real $ n\times n $ matrices. From the definition of the exponential map on matrix groups we obtain:
		\begin{align*}
			T 
			\coloneqq \exp (A + B) - \exp (A) \exp (B)
			&= \sum_{n=0}^{\infty} \frac{(A + B)^{n}}{n!} - \left( \sum_{n=0}^{\infty} \frac{A^{n}}{n!}\right) \left( \sum_{n=0}^{\infty} \frac{B^{n}}{n!}\right)  \\
			&= \sum_{n=0}^{\infty} \frac{(A + B)^{n}}{n!} - \sum_{k=0}^{\infty} \sum_{m=0}^{\infty} \frac{A^{k} B^{m}}{k! m!} \\
			&= \sum_{n=2}^{\infty} \frac{(A + B)^{n}}{n!} - \sum_{k=1}^{\infty} \sum_{m=1}^{\infty} \frac{A^{k} B^{m}}{k! m!} \\
			&= \sum_{n=2}^{\infty} \frac{(A + B)^{n}}{n!} - \left( \sum_{k=1}^{\infty} \frac{A^{k}}{k!} \right) \left( \sum_{m=1}^{\infty} \frac{B^{m}}{m!} \right),
		\end{align*}
		and so
		\begin{equation*}
			\Ltwonorm{T}
			\le \sum_{n=2}^{\infty} \frac{(\Ltwonorm{A} + \Ltwonorm{B})^{n}}{n!} 
			+ \left( \sum_{k=1}^{\infty} \frac{\Ltwonorm{A}^{k}}{k!} \right) \left( \sum_{m=1}^{\infty} \frac{\Ltwonorm{B}^{m}}{m!} \right).
		\end{equation*}
		Substituting $ A = \frac{X}{\lambda} $ and $ B = \frac{Y}{\lambda} $ and using the fact $ \lambda \ge 1 $,
		\begin{align*}
			&\Ltwonorm{\exp \left( \frac{X}{\lambda} + \frac{Y}{\lambda} \right) - \exp \left( \frac{X}{\lambda} \right) \exp \left(\frac{Y}{\lambda}  \right)} \\
			&\phantom{==}\le \frac{1}{\lambda^{2}} \sum_{n=2}^{\infty} \frac{(\Ltwonorm{X} + \Ltwonorm{Y})^{n}}{n!} 
			+ \frac{1}{\lambda^{2}} \left( \sum_{k=1}^{\infty} \frac{\Ltwonorm{X}^{k}}{k!} \right) \left( \sum_{m=1}^{\infty} \frac{\Ltwonorm{X}^{m}}{m!} \right) \\
			&\phantom{==}= \frac{1}{\lambda^{2}}  \left( e^{\Ltwonorm{X} + \Ltwonorm{Y}} - \Ltwonorm{X} - \Ltwonorm{Y} - 1 + \left( e^{\Ltwonorm{X}} - 1 \right) \left( e^{\Ltwonorm{Y}} - 1 \right) \right) \\
			&\phantom{==}= \frac{1}{\lambda^{2}} \left( 2 e^{\Ltwonorm{X} + \Ltwonorm{Y}} - e^{\Ltwonorm{X}} - e^{\Ltwonorm{Y}} - \Ltwonorm{X} - \Ltwonorm{Y} \right).
		\end{align*}
	\end{proof}

	\section{Debiasing the spherical bispectrum estimator}
	\label{DebiasingAppendix}
	Given a sample of $ N $ images as in \eqref{ImageFormationModelRecovery} and their projections $ \left\{ \mathbf{s}_{j}\right\}_{j=1}^{N} $ onto the sphere, our objective is to derive an unbiased estimator $ \widehat{\mathbf{b}} $ of the bispectrum of $ \GrTr $.
	We begin by establishing some notation, making the derivation easier.
	Denote by $ \mathbf{P} = \mathbf{P} \left(L, \mathcal{S}_{t}, \lambda \right) $ the interpolation operator used when projecting an image onto the space of functions on the sphere of bandlimit $ L $. This is a linear operator. 
	Also, recall the definition of $\mathbf{Y} =  \mathbf{Y} \left(L, \mathcal{S}_{t} \right)  $ in \eqref{SphericalHarmonicsMatrix}.

	Taking into account \eqref{ImageFormationModelRecovery}, the spherical harmonics coefficients we estimate for the $ j $th image are of the form 
	\begin{equation}
		\label{SimplerImageFormationModel}
		\mathbf{s}_{j} = \mathbf{U} \mathbf{J}_{j} + \mathbf{U} \varepsilon_{j},\enskip
		\mbox{where }
		\mathbf{U} = \mathbf{Y}^{\dag} \mathbf{P} \mbox{ is a complex matrix},\
		\mathbf{J}_{j} = \mathbf{D} \left(g_{j} \bullet \GrTr \right)
	\end{equation}
	and $\mathbf{s}_{j} $ is the concatenation of all spherical harmonics coefficients of the projection onto the sphere of $ \mathbf{I}_{j} $ arranged in lexicographical order. 
	Note that due to the indexing convention we used for $ \mathbf{Y} $, $ \mathbf{U} $ has rows indexed by $ (\ell,m) $. Since $ \mathbf{P} $ operates on images of $ n \times n $ pixels, the columns are indexed by $ 1,\dots,n^{2} $.
	
	In order to obtain an unbiased estimator, we need to calculate the expected value of $ b_{\mathbf{s}_{j}} \left[\ell_{1}, \ell_{2}, \ell\right] $. 
	As we noted in \Cref{ProjectionSection}, our use of interpolation when projecting an image onto the sphere introduces correlations in the noise.  
	In order to deal with them, we first observe that the spherical bispectrum \eqref{SphericalBispectrum} is multilinear as a function of $ \mathbf{f}_{\ell} $, $ \mathbf{f}_{\ell_{1}}^{*} $ and $ \mathbf{f}_{\ell_{2}}^{*} $, where $ \mathbf{f}_{k} = \left( f_{k,-k}, f_{k,-k+1},\dots,f_{k,k} \right)^{\top} $. 
	In particular, for a fixed $ (\ell_{1}, \ell_{2}, \ell) $ triplet satisfying \eqref{SphericalBispectrumIndices}, we can write $ \mathbf{b}_{\mathbf{s}_{j}} \left[\ell_{1}, \ell_{2}, \ell \right] =  B \left[\mathbf{s}_{j}^{*}, \mathbf{s}_{j}^{*}, \mathbf{s}_{j} \right] $ for some multilinear transformation $ B $. 
	This enables us to write:
	\begin{equation}
		\label{DebiasingStage1}
		\begin{aligned}
			\F[E] \left[ \mathbf{b}_{\mathbf{s}_{j}} [\ell_{1}, \ell_{2}, \ell] \right] 
			&= \F[E] \left[ B\left[ \mathbf{s}_{j }^{*}, \mathbf{s}_{j}^{*}, \mathbf{s}_{j}\right] \right] \\
			&= \F[E] \left[ B \left[ \left(\mathbf{U} \mathbf{J}_{j}\right)^{*}, \left(\mathbf{U} \mathbf{J}_{j}\right)^{*}, \mathbf{U} \mathbf{J}_{j}  \right] \right]
			+ \F[E] \left[ B \left[ \left(\mathbf{U} \bm{\varepsilon}_{j}\right)^{*}, \left(\mathbf{U} \bm{\varepsilon}_{j}\right)^{*}, \mathbf{U} \mathbf{J}_{j} \right] \right] \\
			&\phantom{===} + \F[E] \left[ B \left[ \left(\mathbf{U} \bm{\varepsilon}_{j}\right)^{*}, \left(\mathbf{U} \mathbf{J}_{j}\right)^{*}, \mathbf{U} \bm{\varepsilon}_{j} \right] \right] 
			+ \F[E] \left[ B \left[ \left(\mathbf{U} \mathbf{J}_{j}\right)^{*}, \left(\mathbf{U} \bm{\varepsilon}_{j}\right)^{*}, \mathbf{U} \bm{\varepsilon}_{j} \right] \right]
			\\
			&\phantom{===}+ \F[E] \left[ B \left[ \left(\mathbf{U} \bm{\varepsilon}_{j}\right)^{*}, \left(\mathbf{U} \mathbf{J}_{j}\right)^{*}, \mathbf{U} \mathbf{J}_{j}  \right] \right]
			+ \F[E] \left[ B \left[ \left(\mathbf{U} \mathbf{J}_{j}\right)^{*}, \left(\mathbf{U} \bm{\varepsilon}_{j}\right)^{*}, \mathbf{U} \mathbf{J}_{j}  \right] \right] \\
			&\phantom{===}+ \F[E] \left[ B \left[ \left(\mathbf{U} \mathbf{J}_{j}\right)^{*}, \left(\mathbf{U} \mathbf{J}_{j}\right)^{*}, \mathbf{U} \bm{\varepsilon}_{j}  \right] \right]
			+ \F[E] \left[ B \left[ \left(\mathbf{U} \bm{\varepsilon}_{j}\right)^{*}, \left(\mathbf{U} \bm{\varepsilon}_{j}\right)^{*}, \mathbf{U} \bm{\varepsilon}_{j}  \right] \right]
		\end{aligned}
		%
	\end{equation}
	Consider the form of $ B $ (see \eqref{SphericalBispectrum}) and elements in the sum on the right hand side of \eqref{DebiasingStage1} that depend on either exactly one or exactly three $ \mathbf{U} \bm{\varepsilon}_{j} $ (the last four terms in the sum). 
	Since every element of $ \bm{\varepsilon}_{j}  $ is Gaussian i.i.d. with zero mean, we have
	\begin{equation*}
		\F[E] \left[ \bm{\varepsilon}_{j, i_{1}, i_{2}}  \right] 
		= \F[E] \left[ \bm{\varepsilon}_{j, i_{1}, i_{2}} \bm{\varepsilon}_{j, i_{3}, i_{4}} \bm{\varepsilon}_{j, i_{5}, i_{6}}  \right] 
		= 0.
	\end{equation*}
	Therefore, said elements in the sum in \eqref{DebiasingStage1} are zero and we obtain the following:
	\begin{equation}
		\label{DebiasingStage2}
		\begin{aligned}
			\F[E] \left[ b_{\mathbf{s}_{j}} [\ell_{1}, \ell_{2}, \ell] \right]
			&= \F[E] \left[ B \left[ \left(\mathbf{U} \mathbf{J}_{j}\right)^{*}, \left(\mathbf{U} \mathbf{J}_{j}\right)^{*}, \mathbf{U} \mathbf{J}_{j} \right] \right] 
			+ \F[E] \left[ B \left[ \left(\mathbf{U} \bm{\varepsilon}_{j}\right)^{*}, \left(\mathbf{U} \bm{\varepsilon}_{j}\right)^{*}, \mathbf{U} \mathbf{J}_{j} \right] \right] \\
			&\phantom{===} + \F[E] \left[ B \left[ \left(\mathbf{U} \bm{\varepsilon}_{j}\right)^{*}, \left(\mathbf{U} \mathbf{J}_{j}\right)^{*}, \mathbf{U} \bm{\varepsilon}_{j} \right] \right] 
			+ \F[E] \left[ B \left[ \left(\mathbf{U} \mathbf{J}_{j}\right)^{*}, \left(\mathbf{U} \bm{\varepsilon}_{j}\right)^{*}, \mathbf{U} \bm{\varepsilon}_{j} \right] \right].
		\end{aligned}
	\end{equation}
	Now, denote the elements on the right hand side of \eqref{DebiasingStage2} by $ K_{0} = K_{\mathbf{s}_{j}, 0} \left[\ell_{1}, \ell_{2}, \ell \right]$, $ K_{1} = K_{\mathbf{s}_{j}, 1} \left[\ell_{1}, \ell_{2}, \ell \right]  $, $ K_{2} = K_{\mathbf{s}_{j}, 2} \left[\ell_{1}, \ell_{2}, \ell \right] $ and $ K_{3} = K_{\mathbf{s}_{j}, 3} \left[\ell_{1}, \ell_{2}, \ell \right] $ by order of their appearance. 
	$ K_{0} $  does not depend on the noise. It depends only on $ \mathbf{U} \mathbf{J}_{j} $ the projection onto the sphere of $ \mathbf{J}_{j} $, a discretization of the rotated and translated $ \GrTr $, the ground truth function. We assume the translation is small enough that $ \mathbf{U} \mathbf{J}_{j} $ is approximately $ \kappa_{\lambda} \GrTr $ rotated on the sphere. Thus, 
	\begin{equation}
		\label{UnbiasedTerm}
		K_{0} \approx b_{\kappa_{\lambda} \GrTr} \left[\ell_{1}, \ell_{2} , \ell \right].
	\end{equation} 
	
	In order to handle $ K_{1} $, note that by \eqref{SphericalBispectrum} it has the form 
	\begin{equation}
		\label{K1stage1}
		K_{\mathbf{s}_{j}, 1}
		= \sum_{m=-\ell}^{\ell}  \sum_{m_{1} = -\ell_{1}}^{\ell_{1}} C_{\ell_{1},m_{1}, \ell_{2}, m-m_{1} }^{\ell, m} f_{\ell,m} \F[E] \left[  \left(\mathbf{U} \bm{\varepsilon}_{j}\right)^{*}_{\ell_{1}, m_{1}} \left(\mathbf{U} \bm{\varepsilon}_{j}\right)^{*}_{\ell_{2}, m - m_{1}} \right].
	\end{equation}
	Using the fact $ \bm{\varepsilon}_{j, k, m} \sim \mathcal{N} (0,\sigma^{2}) $ i.i.d., we get
	\begin{align*}
		\F[E] \left[  \left(\mathbf{U} \bm{\varepsilon}_{j}\right)^{*}_{\ell_{1}, m_{1}} \left(\mathbf{U} \bm{\varepsilon}_{j}\right)^{*}_{\ell_{2}, m - m_{1}} \right]
		&= \F[E] \left[  
		\left( \sum_{\mathbf{k}_{1}} \mathbf{U}_{(\ell_{1}, m_{1}), \mathbf{k}_{1}}^{*} \bm{\varepsilon}_{j,\mathbf{k}_{1}} \right) 
		\left(\sum_{\mathbf{k}_{2}} \mathbf{U}_{(\ell_{2}, m-m_{1}), \mathbf{k}_{2}}^{*} \bm{\varepsilon}_{j,\mathbf{k}_{2}}   \right)
		\right] \\
		&= \sum_{\mathbf{k}_{1}} \sum_{\mathbf{k}_{2}} 
		\mathbf{U}_{(\ell_{1}, m_{1}), \mathbf{k}_{1}}^{*}  \mathbf{U}_{(\ell_{2}, m-m_{1}), \mathbf{k}_{2}}^{*} 
		\F[E] \left[ \bm{\varepsilon}_{j,\mathbf{k}_{1}} \bm{\varepsilon}_{j,\mathbf{k}_{2}}  \right] \\
		&= \sum_{\mathbf{k}_{1}} \sum_{\mathbf{k}_{2}} 
		\mathbf{U}_{(\ell_{1}, m_{1}), \mathbf{k}_{1}}^{*}  \mathbf{U}_{(\ell_{2}, m-m_{1}), \mathbf{k}_{2}}^{*} 
		\cdot \sigma^{2} \delta_{\mathbf{k}_{1}, \mathbf{k}_{2}} \\
		&= \sigma^{2}\sum_{\mathbf{k}_{2}} 
		\mathbf{U}_{(\ell_{1}, m_{1}), \mathbf{k}}^{*}  \mathbf{U}_{(\ell_{2}, m-m_{1}), \mathbf{k}}^{*}  \\
		&= \sigma^{2} \left( \mathbf{U} \mathbf{U}^{\top} \right)_{(\ell_{1}, m_{1}),(\ell_{2}, m-m_{1})}^{*}.
	\end{align*}
	Substituting that into \eqref{K1stage1}, it follows that 
	\begin{equation}
		\label{K1explicit}
		K_{1} 
		=\sigma^{2} 
		\sum_{m=-\ell}^{\ell} \mathbf{s}_{j, (\ell,m)} 
		\sum_{m_{1} = \max\left\{ -\ell_{1}, m - \ell\right\}}^{\min\left\{ \ell_{1}, m+\ell_{2} \right\}}
		C_{\ell_{1}, m_{1}, \ell_{2}, m-m_{1} }^{\ell,m}
		\left( \mathbf{U} \mathbf{U}^{\top} \right)_{(\ell_{1}, m_{1}),(\ell_{2}, m-m_{1})}^{*} .
	\end{equation}
	Using a similar approach and remembering that the Clebsch-Gordan coefficients are non-zero only if $ m_{1} + m_{2} = m $, one can show that 
	\begin{align}
		\label{K2explicit}
		K_{2} 
		&= \sigma^{2} 
		\sum_{m_{2}= -\ell_{2}}^{\ell_{2}} \mathbf{s}_{j, (\ell_{2}, m_{2})}^{*}  
		\sum_{m=\max\left\{ -\ell, m_{1} - \ell_{1} \right\}}^{\min\left\{ \ell, m+\ell_{2} \right\}}
		C_{\ell_{1}, m-m_{2}, \ell_{2}, m_{2} }^{\ell,m}
		\left( \mathbf{U} \mathbf{U}^{\dag} \right)_{(\ell,m), (\ell_{1}, m_{1})} \\
		\label{K3explicit}
		K_{3} 
		&= \sigma^{2} 
		\sum_{m_{1}= -\ell_{1}}^{\ell_{1}} \mathbf{s}_{j, (\ell_{1}, m_{1})}^{*} 
		\sum_{m=\max\left\{ -\ell, m_{1} - \ell_{1} \right\}}^{\min\left\{ \ell, m+\ell_{2} \right\}}
		C_{\ell_{1}, m_{1}, \ell_{2}, m-m_{1} }^{\ell,m}
		\left( \mathbf{U} \mathbf{U}^{\dag} \right)_{(\ell,m), (\ell_{2}, m_{1})}.
	\end{align}
	
	Combining \eqref{DebiasingStage2} and \eqref{UnbiasedTerm} and using the explicit expressions of $ K_{1} $, $ K_{2} $ and $ K_{3} $ in \eqref{K1explicit}, \eqref{K2explicit} and \eqref{K3explicit}, we obtain an unbiased estimator of $ \widehat{\mathbf{b}} \left[\ell_{1}, \ell_{2}, \ell\right] $:
	\begin{equation*}
		\widehat{\mathbf{b}} \left[\ell_{1}, \ell_{2}, \ell \right]
		= \frac{1}{N} \sum_{j=1}^{N} \left( \mathbf{b}_{\mathbf{s}_{j}} \left[\ell_{1}, \ell_{2}, \ell\right] - \sum_{k=1}^{3} K_{\mathbf{s}_{j}, k} \left[\ell_{1}, \ell_{2}, \ell  \right]\right) .
	\end{equation*}
	
	When implementing this estimator, it is important to note that $ K_{\mathbf{s}_{j}, k} $  is a linear transformation, when applied to the realized version of $ \mathbf{s}_{j} $. It is also sparse. Implementing it as such enables one to easily implement this estimator as a simple matrix multiplication and can considerably speed up its calculation.
	
	\section{Noise statistics of projected images}
	\label{NoiseStatisticsAppendix}
	We wished to study the effects of the projection onto the sphere on noise. 
	Specifically, given an image with additive white noise, $ \mathrm{Y} = \mathrm{I} + \bm{\varepsilon} $, we expect the spherical harmonics coefficients to have additive noise, because the projection onto the sphere is linear. 
	We sought to test empirically whether the noise on the sphere can still be regarded as white noise.
	
	To that end, we generated $ 10^{4} $ white noise images of size $ 101 \times 101 $. Specifically, each pixel was Gaussian i.i.d. with zero mean and variance $ 1 $.
	The power spectrum of every image was calculated and averaged over the corresponding pixels. 
	We then projected every image onto the sphere as described in \Cref{ProjectionSection}, estimating the spherical harmonics coefficients of the projection up to bandlimit $ 50 $. 
	Then, we calculated the normalized power spectrum of every projected image and averaged them.
	Note that the normalized power spectrum $ P $ of a function $ f $ on the sphere for the $ \ell $th frequency is given by 
	\begin{equation}
		\label{eq:PowerSpectrum}
		P \left[ f\right]_{\ell}
		= \frac{1}{2\ell + 1} \sum_{m = -\ell}^{\ell} \left|f_{\ell, m}\right|^{2}.
	\end{equation}
	
	The results of this experiment are shown in \autoref{fig:PowerSpectrumTest}. 
	The average power spectrum of the original images and the normalized power spectrum on the sphere behaved as expected from the power spectra of white noise. In both cases, the power spectrum was approximately constant over all frequencies.
	
	\begin{figure}
		\begin{center}
			\includegraphics[scale=0.6]{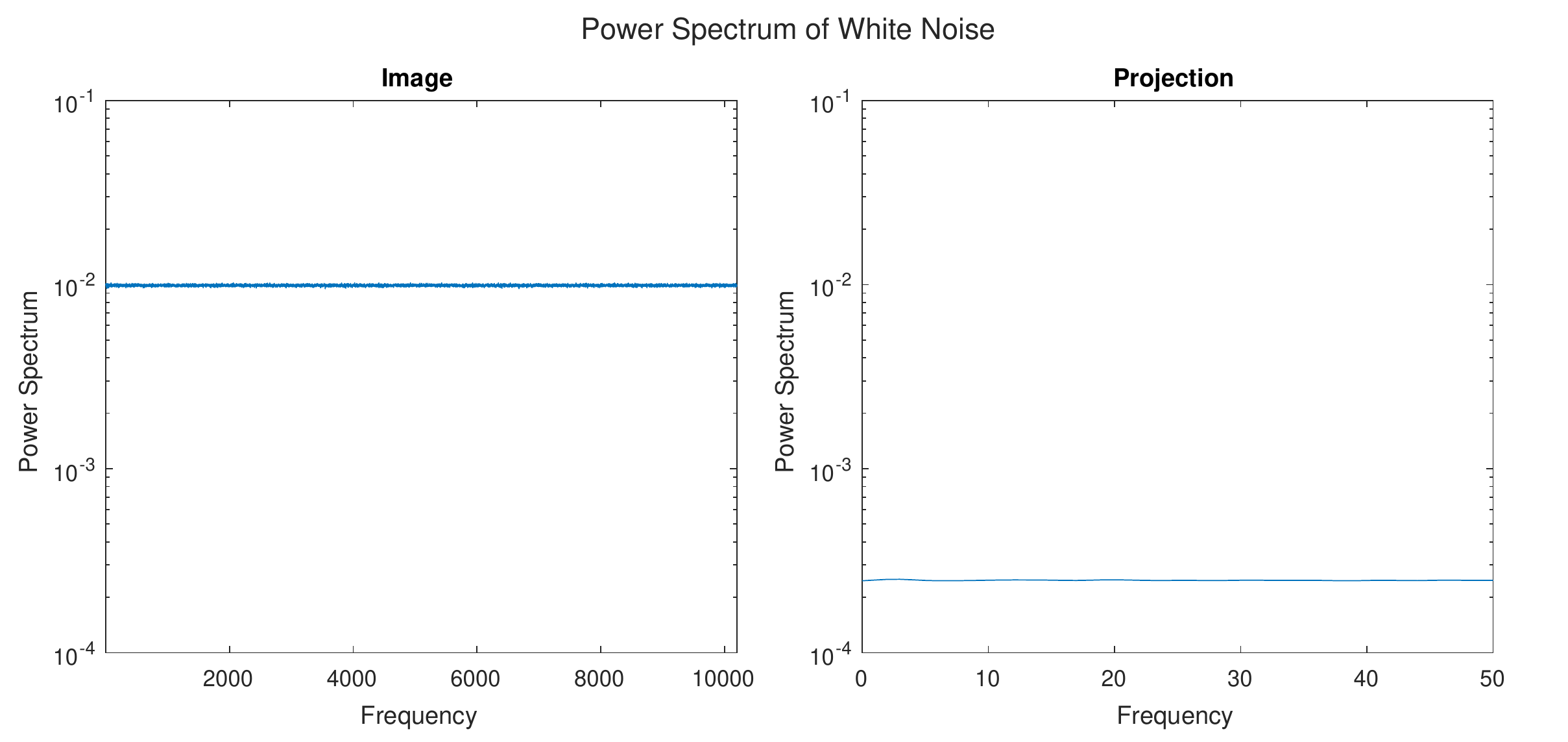}
		\end{center}
		\caption{LEFT: the power spectrum of a white noise image, averaged over $ 10^4 $ noise realizations. The frequencies are ordered lexicographically.
			RIGHT: the power spectrum \eqref{eq:PowerSpectrum} of the corresponding projection of the image up to bandlimit $ 50 $, averaged over the same $ 10^{4} $ noise realizations projected onto the sphere.}
		\label{fig:PowerSpectrumTest}
	\end{figure}

	\section*{Acknowledgments}
	We thank Nicolas Boumal for his major contribution to this project's conception and first steps; this project is also his brainchild.
	This work was funded by NSF-BSF award 2019752. 
	N.S is also partially supported by the BSF grant no. 2018230. 
	T.B. is also partially supported by the Zimin Institute for Engineering Solutions Advancing Better Lives, the BSF grant no. 2020159, and the ISF grant no. 1924/21.
	
	\printbibliography
	
\end{document}